\documentclass[%
reprint,
superscriptaddress,
 amsmath,amssymb,
 aps,
]{revtex4-1}

\usepackage{mathtools}
\usepackage{dcolumn}
\usepackage{bm}

\usepackage{graphicx}
\usepackage{float}
\usepackage{subcaption}

\usepackage{amsmath}
\usepackage{amssymb}
\usepackage{amsthm}
\usepackage{algorithmic}
\usepackage{mathrsfs}
\usepackage{braket}
\usepackage{xcolor}
\usepackage{tikz}
\usepackage[roman]{complexity}
\usepackage{pax}

\newcommand{\wt}{\widetilde}

\usepackage{soul}
\newcommand{\risi}[1]{{\color{blue}#1}}

\newtheorem*{theorem*}{Theorem}
\newtheorem{theorem}{Theorem}
\newtheorem{lemma}{Lemma}
\newtheorem{corollary}{Corollary}
\newtheorem{definition}{Definition}
\newtheorem*{definition*}{Definition}

\usepackage{pdfpages}
\usepackage{pgffor}

\makeatletter
\AtBeginDocument{\let\LS@rot\@undefined}
\makeatother

\usepackage[hyperindex]{hyperref}
\hypersetup{colorlinks,linkcolor={blue},citecolor={blue},urlcolor={black}} %

\begin{document}

\title{Speeding up Learning Quantum States through Group Equivariant Convolutional Quantum Ans{\"a}tze}

\author{Han Zheng}
\email{hanz98@uchicago.edu}
\affiliation{Department of Statistics, The University of Chicago, Chicago, IL 60637, USA}
\affiliation{DAMTP, Center for Mathematical Sciences, University of Cambridge, Cambridge CB30WA, UK}

\author{Zimu Li}
\email{lizm@mail.sustech.edu.cn}
\affiliation{DAMTP, Center for Mathematical Sciences, University of Cambridge, Cambridge CB30WA, UK}

\author{Junyu Liu$^*$}
\email{junyuliu@uchicago.edu}
\affiliation{Pritzker School of Molecular Engineering, The University of Chicago, Chicago, IL 60637, USA}
\affiliation{Chicago Quantum Exchange, Chicago, IL 60637, USA}
\affiliation{Kadanoff Center for Theoretical Physics, The University of Chicago, Chicago, IL 60637, USA}

\author{Sergii Strelchuk}
\email{ss870@cam.ac.uk}
\affiliation{DAMTP, Center for Mathematical Sciences, University of Cambridge, Cambridge CB30WA, UK}

\author{Risi Kondor}
\email{risi@cs.uchicago.edu}
\affiliation{Department of Statistics, The University of Chicago, Chicago, IL 60637, USA}
\affiliation{Department of Computer Science, The University of Chicago, Chicago, IL 60637, USA}
\affiliation{Flatiron Institute, New
York City, NY 10010, USA}

\date{\today}

\begin{abstract}
    We develop a theoretical framework for $S_n$-equivariant convolutional quantum circuits with SU$(d)$-symmetry, building on and significantly generalizing Jordan's Permutational Quantum Computing (PQC) formalism based on Schur-Weyl duality connecting both SU$(d)$ and $S_n$ actions on qudits. In particular, we utilize the Okounkov-Vershik approach to prove Harrow's statement on the equivalence between $\operatorname{SU}(d)$ and $S_n$ irrep bases and to establish the $S_n$-equivariant Convolutional Quantum Alternating Ans{\"a}tze ($S_n$-CQA) using Young-Jucys-Murphy (YJM) elements. We prove that $S_n$-CQA is able to generate any unitary in any given $S_n$ irrep sector, which may serve as a universal model for a wide array of quantum machine learning problems with the presence of SU($d$) symmetry. Our method provides another way to prove the universality of Quantum Approximate Optimization Algorithm (QAOA) and verifies that 4-local SU($d$) symmetric unitaries are sufficient to build generic SU($d$) symmetric quantum circuits up to relative phase factors. We present numerical simulations to showcase the effectiveness of the ans{\"a}tze to find the ground state energy of the $J_1$--$J_2$ antiferromagnetic Heisenberg model on the rectangular and Kagome lattices. Our work provides the first application of the celebrated Okounkov-Vershik's $S_n$ representation theory to quantum physics and machine learning, from which to propose quantum variational ans\"atze that strongly suggests to be classically intractable tailored towards a specific optimization problem. 
\end{abstract}

\maketitle
$*$: corresponding author.


\section{Introduction}

The combination of new ideas from machine learning and recent developments in quantum computing has lead to an impressive array of new applications. \cite{harrow2009quantum,wiebe2012quantum,lloyd2014quantum,wittek2014quantum,wiebe2014quantum,rebentrost2014quantum,biamonte2017quantum,mcclean2018barren,schuld2019quantum,tang2019quantum,havlivcek2019supervised,liu2021rigorous,Liu:2021wqr}. Prominent examples of this interplay are the Variational Quantum Eigensolver (VQE) and Quantum Approximate Optimization Algorithm (QAOA) \cite{Farhi_2014,mcclean2016theory,cerezo2021variational}, which are considered among the most promising quantum machine learning approaches in the Noisy Intermediate-Scale Quantum (NISQ) \cite{preskill2018quantum} era. VQEs and QAOA have shown tremendous promise in quantum simulation and quantum optimization \cite{peruzzo2014variational,Hadfield_2019,mcardle2020quantum,Yuan:2020xmq,Liu:2021otn, Farhi_2021}.

One of the most important neural network architectures in classical machine learning are Convolutional Neural Networks (CNNs) \cite{lecun1998gradient, krizhevsky2012imagenet,simonyan2014very,szegedy2015going,lecun2015deep}. In recent years, CNNs have also found applications in condensed matter physics and quantum computing. For instance, \cite{Lukin_2019} proposes a quantum convolutional neural network with $\log N$ parameters to solve topological symmetry-protected phases in quantum many-body systems, where $N$ is the system size. One of the key properties of classical CNNs is equivariance, which roughly states that if the input to the neural network is shifted, then its activations translate accordingly. There have been several attempts to introduce theoretically sound analogs of convolution and equivariance to quantum circuits, but they have generally been somewhat heuristic. The major difficulty is that the translation invariance of CNNs lacks a mathematically rigorous quantum counterpart due to the discrete spectrum of spin-based quantum circuits. For example, \cite{Lukin_2019} uses the quasi-local unitary operators to act vertically across all qubits.

In quantum systems there is a discrete set of translations corresponding to permuting the qudits as well as a continuous notion of translation corresponding to spatial rotations by elements of $\operatorname{SU}(d)$. Combining these two is the realm of so-called Permutational Quantum Computing (PQC) \cite{pqc}. Therefore, a natural starting point for realizing convolutional neural networks in quantum circuits is to look for {\it permutation equivariance}. In one of our related works \cite{Zheng2022PQC}, we argued that the natural form of equivariance in quantum circuits is permutation equivariance and we introduced a theoretical framework to incorporate group-theoretical CNNs into the quantum circuits, building on and generalizing the PQC framework to what we call PQC+~\cite{Zheng2022PQC}. 

In this paper, we further explore PQC+ and its significance for machine learning applications. Roughly speaking, PQC+ machine consists of unitary time evolutions of $k$-local SU($d$)-symmetric Hamiltonian. As a feature, \emph{Schur-Weyl duality} between the aforementioned $S_n$ and $\operatorname{SU}(d)$ actions on qudits systems appears naturally and will be used throughout the paper. Most importantly, it indicates that any $\operatorname{SU}(d)$ symmetric quantum circuits can be expressed in $S_n$ irreducible representations (irreps). Exploiting the power of $S_n$ representation theory in quantum circuits towards NISQ applications is thus the central theme of the paper. The representation theory of $S_n$  has been found to be a powerful tool in various  permutation equivariant learning tasks, e.g., learning set-valued functions \cite{Zaheer2017a} and learning on graphs \cite{maron2018invariant,thiede2020general}. Most applications of permutation-equivariant neural networks work with a subset of representations of $S_n$. In contrast, in physical and chemical models where the Hamiltonian exhibits global $\operatorname{SU}(d)$ symmetry, such as the Heisenberg model, it is necessary to consider \emph{all} the $S_n$ irreps (a detailed explanation of this significant insight can be found in Section \ref{s5}). However, even the best classical Fast Fourier Transforms (FFTs) over the symmetric group $S_n$ require at least $\mathcal{O}(n!n^2)$ operations \cite{Clausen_1993,Maslen_1998},
which dashes any hope of calculating the Fourier coefficients even for relatively small $n$. Indeed, despite increasing realization of the importance of enforcing $\operatorname{SU}(2)$ symmetry, none of the neural-network quantum state (NQS) ans{\"a}tze are able to respect $\operatorname{SU}(2)$ symmetry for all $\operatorname{SU}(2)$ irreps, due to the super-polynomial growth of the multiplicities of irreps and the super-exponential cost to compute Fourier coefficients over $S_n$. Finding variational ans{\"a}tze respecting continuous rotation symmetry is desirable because it not only helps to gain important physical insights about the system but also leads to more efficient simulation algorithms \cite{Vieijra_2021, VieijraRBM}.   

Motivated by the class of problems with a global $\operatorname{SU}(d)$ symmetry, in Section \ref{s3}, we construct what we call the variational $S_n$-equivariant Convolutional Quantum Alternating ans\"{a}tze ($S_n$-CQA), which are products of alternating exponentials of certain Hamiltonians admitting SU(d) symmetry. This is a concrete example of the PQC+ framework and may also be thought as a special case of QAOA with SU($d$) symmetry. Using the \emph{Okounkov-Vershik approach} \cite{Okounkov1996} to $S_n$ representation theory as well as other classical results from the theory of Lie group and Lie algebra \cite{Kuranishi51,Yamabe1950} we prove that $S_n$-CQA generates any unitary matrix in each given $S_n$-irrep block decomposed from an $n$-qudit system, hence it acts as a \emph{restricted universal} variational model for problems that possess global SU($d$) symmetry (Theorem \ref{ex2}). Consequently, it can be applied to a wide array of machine learning and optimization tasks that exhibit global SU($d$) symmetry or require explicit computation of high dimensional $S_n$-irreps, presenting a quantum super-exponential speed-up. Our proof techniques are of independent interest and we provide two more applications. It is shown in \cite{Lloyd_2018, Morales_2020} that QAOA ans\"{a}tze generated by simple local Hamiltonians are universal in the common sense. Forgetting the imposed symmetry, we use our techniques to derive the universality for a different but related class of QAOA ans{\"a}tze with a richer set of mixer Hamiltonians (Theorem \ref{QAOA-Univ}). In addition, we find a \emph{4-local} $S_n$-CQA model which is universal to build any SU($d$) symmetric quantum circuits up to phase factors (Theorem \ref{4-local}) with awareness of the fact that \emph{2-local} SU(d) symmetric unitaries cannot fulfill the task when $d \geq 3$ \cite{Marin1,MarvianSUd}. Consequently, when compared with other SU(d) symmetric ans{\"a}tze, products of exponentials of SWAPs (eSWAPs) proposed in \cite{Seki_2020} admit the restricted universality with $\operatorname{SU(d)}$ symmetry only when $d = 2$ and the CQA model is universal in general cases when restricted to any one of $S_n$ irreps. 

In Section \ref{s4} we explore more details about Schur-Weyl duality on qudits systems. To be specific, talking about $S_n$ or SU($d$) irrep blocks in a qudits system requires using the \emph{Schur basis}, instead of the computational basis. The Schur basis can be constructed by either SU($d$) Clebsch-Gordan decomposition \cite{harrow2005PhD, pqc, sergii1} or by the $S_n$ branching rule \cite{Goodman2009,Tolli2009,krovi}, which yields two ways to label the basis elements by either SU($d$) Casimir operators or by the so-called \emph{YJM-elements} used in Okounkov-Vershik approach. We rigorously demonstrate the equivalence of these labeling schemes (Theorem \ref{Spin-Content}), first conjectured in Harrow's thesis \cite{harrow2005PhD}, (see also discussions by Bacon, Chuang, and Harrow \cite{Harrow06} and Krovi \cite{krovi}). As a result, we find a state initialization method, using constant-depth qudit circuits, to produce linear combinations of Schur basis vectors, which may be preferred in NISQ devices rather than implementing a Quantum Schur transform (QST). We show that the measurements taken for variationally updating parameters in $S_n$-CQA can be efficiently calculated on the Schur basis, while similar conclusion is unlikely to be drawn classically. 

In the numerical part~\ref{s5}, we illustrate the potential of this framework by applying it to the problem of finding the ground state energy of $J_1-J_2$ antiferromagnetic Heisenberg magnets, a gapless system with no known sign structure or analytical solution in quantum many body theory. We compare our model with classical and quantum algorithms like \cite{VieijraRBM,Roth_2021,Seki_2020}. We emphasize the consequence of the failure of \textit{Marshall-Lieb-Mattis theorem} \cite{Marshall1955,Lieb62,Tasaki2020} in the frustrated region with which classical neural networks struggle to discover the sign structure to an admissible accuracy due to violation of global $\operatorname{SU}(2)$ symmetry \cite{Choo_2019,Castelnovo_2020}. We include numerical simulation to show the effectiveness of the $S_n$-CQA ans{\"a}tze in finding the ground state with frustration using only $\mathcal{O}(pn^2)$ parameters for $p$ alternating layers. Noisy simulations are also provided to show the robustness of $S_n$-CQA. 

Our theoretical results about SU($d$) symmetry can be reversed to exhibit $S_n$ permutation symmetry. We define SU($d$)-CQA on SU($d$)-irrep blocks and leave it to future research work to explore its theoretical and experimental potential. All statements and theorems discussed in the main text are proved in full detail in the Supplementary Materials (SM).


\section{Background on Representation Theory of The Symmetric Group} \label{s2}

In this section we define some of the mathematical concepts and notations used in the rest of the paper. Further details can be found in~\cite{Sagan01,Goodman2009,Tolli2009}.

Let $V$ be a $d$-dimensional complex Hilbert space with orthonormal basis $\{e_1,...,e_d\}$. The tensor product space $V^{\otimes n}$ admits two natural representations: 
the \textit{tensor product representation} $\pi_{\operatorname{SU}(d)}$ of $\operatorname{SU}(d)$ acting as
\begin{align*}
	\pi_{\operatorname{SU}(d)}(g) (e_{i_1} \otimes \cdots \otimes e_{i_n}) \vcentcolon = g \cdot e_{i_1} \otimes \cdots \otimes g \cdot e_{i_n}, 
\end{align*}
where $g \cdot e_{i_k}$ is the fundamental representation of $\operatorname{SU}(d)$, and the \textit{permutation representation} $\pi_{S_n}$ of $S_n$ acting as
\begin{align*}
	\pi_{S_n}(\sigma) (e_{i_1} \otimes \cdots \otimes e_{i_n}) \vcentcolon =
	e_{i_{\sigma^{-1}(1)}} \otimes \cdots \otimes e_{i_{\sigma^{-1}(n)}}.
\end{align*}
We treat $V^{\otimes n}$ as the Hilbert space of an $n$-qudit system. The so-called Schur--Weyl duality reveals how the above two representations are related.

Schur--Weyl duality is widely used in quantum computing \cite{childs2007weak}, quantum information theory \cite{harrow2005PhD} and high energy physics \cite{Keppeler2018}. In particular, in Quantum Chromodynamics it was used to decompose the $n$-fold tensor product of $\operatorname{SU}(3)$ representations. In that context, standard Young tableaux are referred to as Weyl-tableaux and labeled by the three iso-spin numbers ($u,d,s$). The underlying Young diagrams containing three rows $\lambda = (\lambda_1,\lambda_2,\lambda_3)$ are used to denote an $\operatorname{SU}(d)$ irreducible representation (irrep). There is another way using Young diagrams $\lambda' = \lambda_1 - \lambda_3,\lambda_2 - \lambda_3)$ labeled by Dynkin integers via highest weight vectors. In short, there are two conventions in literature to denote $\operatorname{SU}(d)$ irreps. On the other hand, $S_n$ irreps can also be denoted by Young diagrams \cite{Sagan01}. Schur--Weyl duality says that irreps of $\operatorname{SU}(d)$ and $S_n$ are dual in the following sense and denoted by the same Young diagrams \emph{with $n$ boxes and at most $d$ rows}.

\begin{figure}[H]
	\centering
	\includegraphics[width=2in]{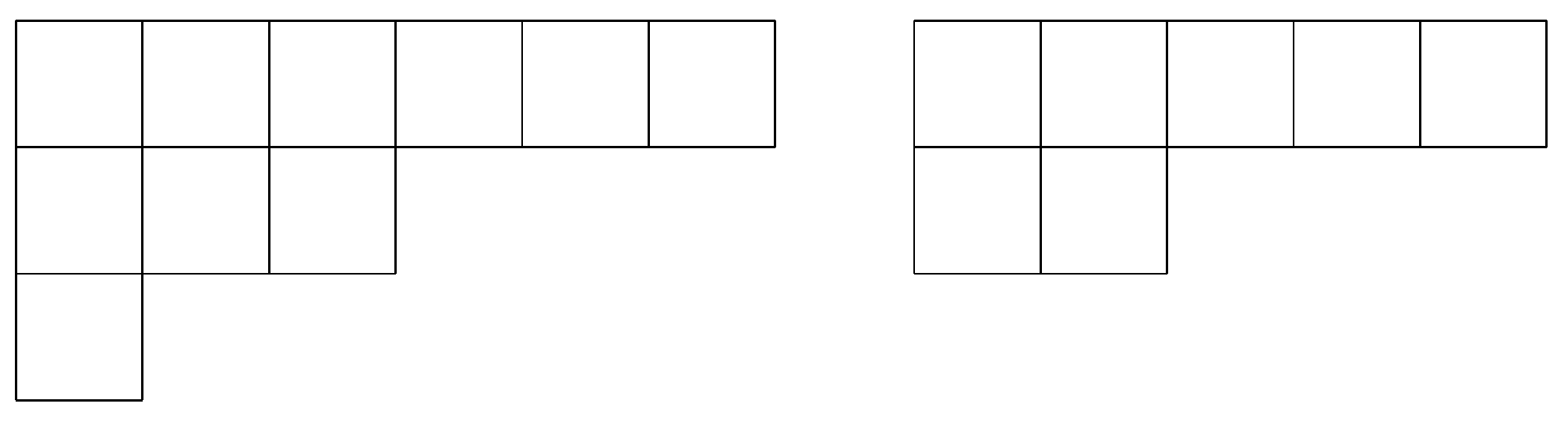}
	\caption{Young diagrams $\lambda = (6,3,1)$ and $\lambda' = (5,2)$ represent to the same $\operatorname{SU}(3)$-irrep of highest weight $(3,2)$. However, in the context of Schur--Weyl duality, $\lambda$ corresponds to an $S_{10}$ irrep while $\lambda'$ gives an $S_{7}$ irrep.}
	\label{YoungDiagram}
\end{figure}

\begin{theorem*}[Schur--Weyl Duality] 
	The action of $\operatorname{SU}(d)$ and $S_n$ on $V^{\otimes n}$ jointly decompose the space into irreducible representations of both groups in the form 
	\begin{align*}
		V^{\otimes n} = \bigoplus_\lambda W_\lambda\otimes S^\lambda,
	\end{align*}
	where $W_\lambda$ and $S^\lambda$ denote irreps of $\operatorname{SU}(d)$  resp.~$S_n$, and $\lambda$ ranges over all Young diagrams of size $n$ with at most $d$ rows. Consequently,
	\begin{align*}
		& \pi_{\operatorname{SU}(d)} \cong \bigoplus_\mu W_\mu \otimes \operatorname{1}_{m_{\operatorname{SU}(d),\mu}}, \quad  \pi_{S_n} \cong \bigoplus_\lambda \operatorname{1}_{ m_{S_n,\lambda}} \otimes S^\lambda,
	\end{align*}
	where $m_{\operatorname{SU}(d),\mu}=\dim S^\mu$ and $m_{S_n,\lambda}=\dim W_\lambda$.
	
\end{theorem*}

\begin{figure*}[ht]
	\centering
	\includegraphics[width=0.9\textwidth]{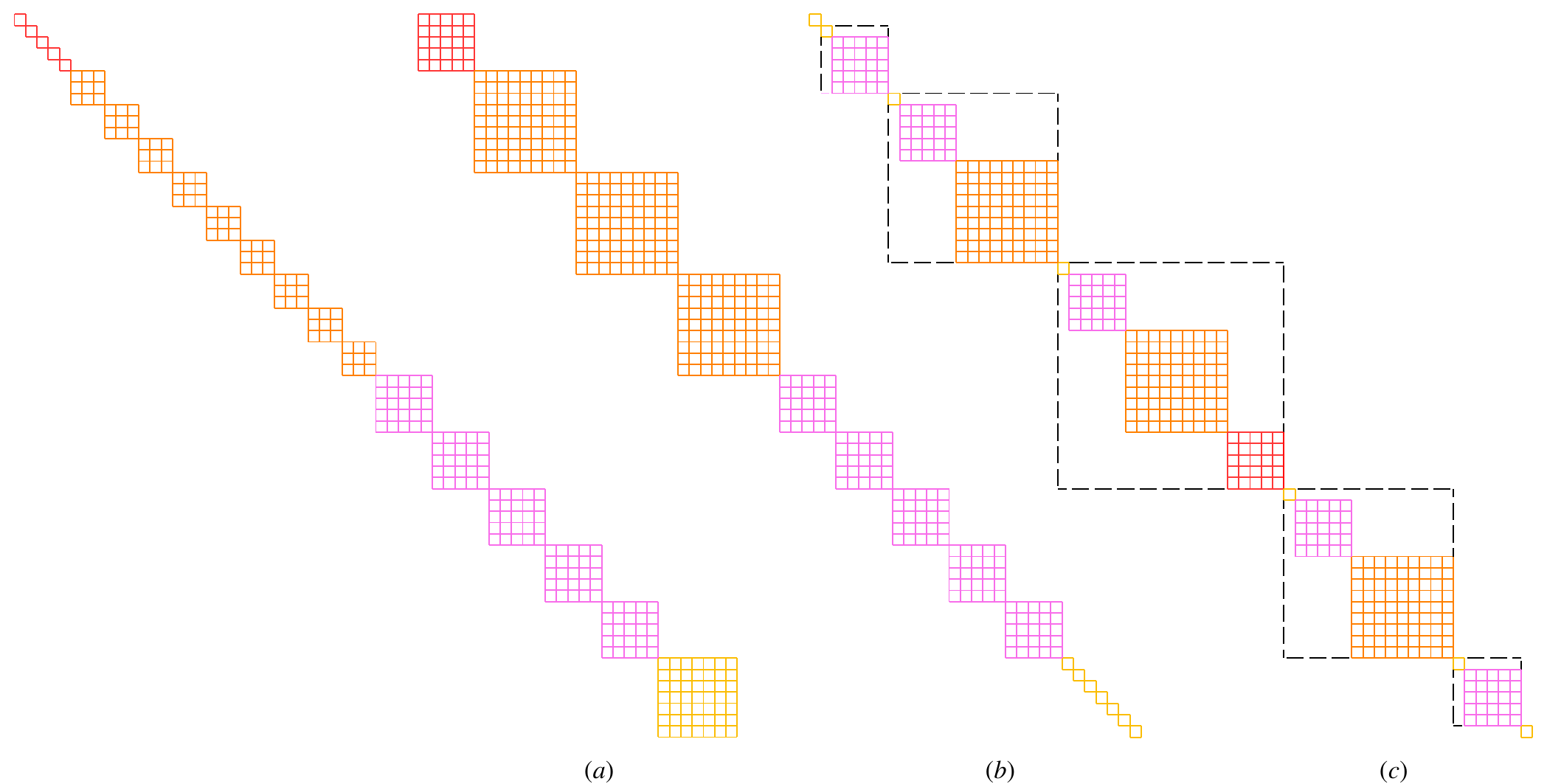}
	\caption{$(a)$ is the decomposition of $(\mathbb{C}^2)^6$ with respect to $\operatorname{SU}(2)$ action while $(b)$ is for $S_6$ by Schur-Weyl duality. $(c)$ is an arrangement of $(b)$ which respects both permutation modules and $S_6$ irreps.}
	\label{SchurWeyl}
\end{figure*}

One can easily verify that $\pi_{\operatorname{SU}(d)}$ and $\pi_{S_n}$ commute (further properties are described in the SM). Consider the \textit{symmetric group algebra} $\mathbb{C}[S_n]$ consisting of all formal finite sums $f = \sum_i c_i \sigma_i$. Its representation is then $\tilde{\pi}_{S_n}(f) = \sum_i c_i \pi_{S_n}(\sigma_i)$. When there is no ambiguity, we denote by $U_\sigma$ or simply $\sigma$ the representations $\pi_{S_n}(\sigma)$. 

Working from the perspective of Schur--Weyl duality requires, at least theoretically, using the \emph{Schur basis} rather than the computational basis. A conventional way to build such a basis is conducting sequential coupling and Clebsch--Gordan decompositions of $\operatorname{SU}(d)$ representations \cite{harrow2005PhD, pqc, sergii1} which transform common matrix representations of $\operatorname{SU}(d)$ into irrep matrix blocks like in Fig.\ref{SchurWeyl} $(a)$. Since our focus are ans{\"a}tze, operators and quantum circuits with $\operatorname{SU}(d)$ symmetry which commute with $\pi_{\operatorname{SU}(d)}$, and since Schur--Weyl duality and the double commutant theorem (see SM) say that they must be established from the group algebra $\mathbb{C}[S_n]$, we need to explore $S_n$ irreps blocks as in Fig.\ref{SchurWeyl} $(b)$ however. We are going to introduce a method to decompose permutation matrices in this picture and explain basic notions in $S_n$ representation theory and the Okounkov-Vershik approach, since they are essential to understand the theoretical results in this paper. 

We first consider the so-called \textit{permutation module} $M^\mu$. In the case of qubits, permutation modules correspond to sets of Schur basis elements having different read-out on total spin components. Fortunately there is an accessible way to understand $M^\mu$ in the tensor product space $V^{\otimes n}$. To make things simpler, consider the $d=2$ case of $\operatorname{SU}(2)-S_n$ duality on $(\mathbb{C}^2)^{\otimes n}$. Only two-row Young diagrams $\lambda = (\lambda_1, \lambda_2)$ appear in this duality and the half of difference $\frac{1}{2}(\lambda_1 - \lambda_2)$ between the lengths of the two rows gives the total spin of the $\operatorname{SU}(2)$-irrep $W_\lambda$. The permutation module $M^\mu$ is isomorphic with the linear span of all computational basis vectors with $z$-spin components equal to $\frac{1}{2}(\mu_1 - \mu_2)$.

Note that, $(\mathbb{C}^2)^{\otimes n} = \bigoplus_\mu M^\mu$ and each $M^\mu$ is invariant under $S_n$ permutation. Furthermore, $M^\mu$ can be further decomposed into $S_n$-irreps. For two-row Young diagrams ($\operatorname{SU}(2)-S_n$ duality), the decomposition is easy: $M^\mu = \bigoplus_{\lambda \geq \mu} S^\lambda$ where $\lambda, \mu$ have the same size $n$ and we use the \textit{dominance order} $\lambda \unrhd \mu$ if $\lambda_1 \geq \mu_1$. In summary, we have
\begin{align*}
	(\mathbb{C}^2)^{\otimes n} = \bigoplus_\mu M^\mu = \bigoplus_\mu \bigoplus_{\lambda \unrhd \mu} S^\lambda \cong \bigoplus_\lambda \operatorname{1}_{ m_{S_n,\lambda}} \otimes S^\lambda. 
\end{align*} 
Isomorphic copies of $S^\lambda$ come from different permutation modules.  The largest permutation module contains all distinct $S_n$-irreps in $(\mathbb{C}^2)^{\otimes n}$ (see e.g., Fig.\ref{SchurWeyl} $(c)$). For general Young diagrams, decompositions of $M^\mu$ would have nontrivial multiplicities \cite{Sagan01,Tolli2009}.

\begin{figure}[H]
	\centering
	\includegraphics[width=3.3in]{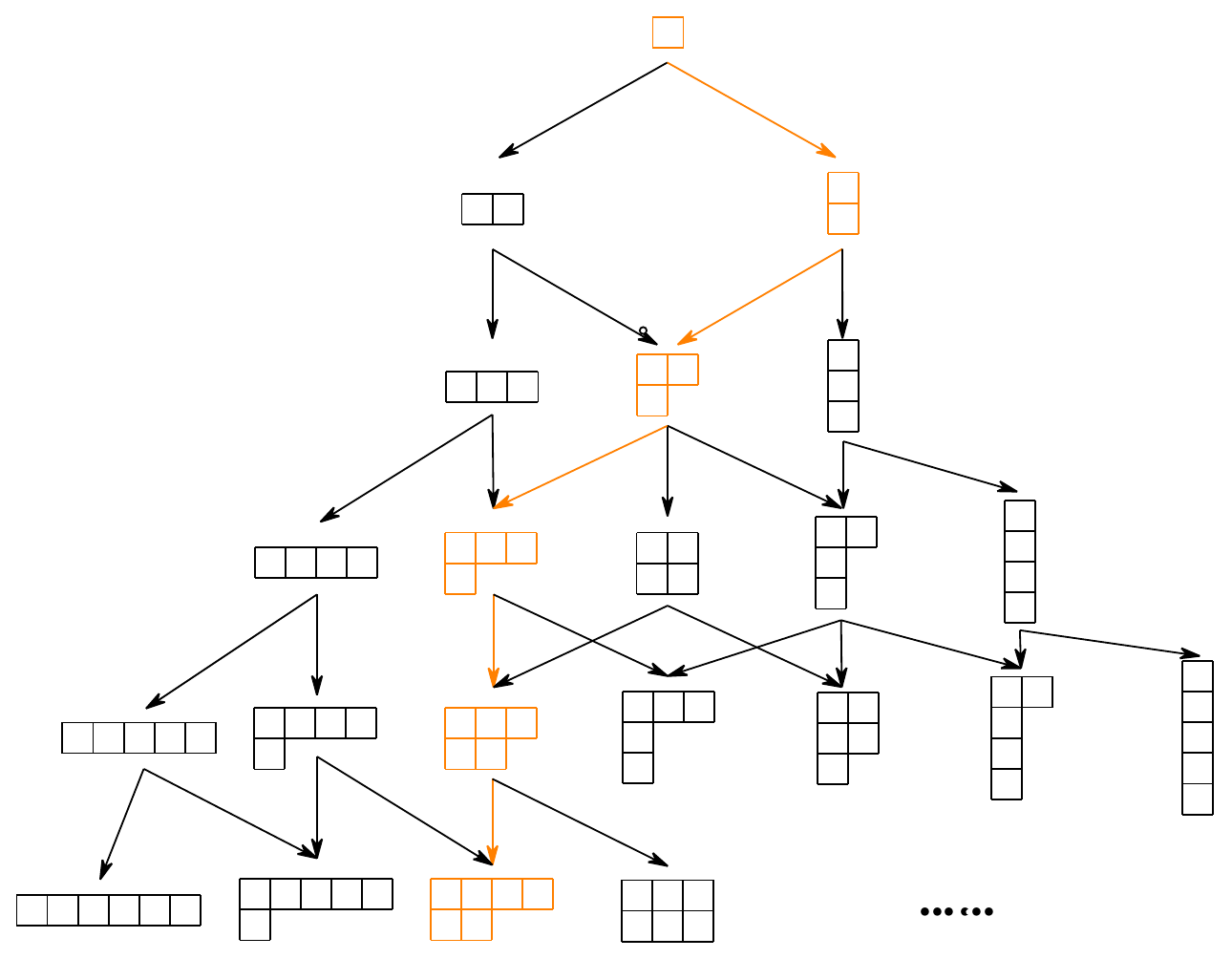}
	\caption{Bratteli diagram for $S_6$. Upper Young diagrams connecting by arrows to lower ones arise from the decomposition. Orange arrows form a path. Note that diagrams with more than two rows cannot appear in $\operatorname{SU}(2)-S_6$ duality.}
	\label{Bratteli}
\end{figure}

Each $S^\lambda$ can be decomposed further with respect to $S_{n-1} \subset S_n$ as $S^\lambda = \bigoplus_{\rho } S^{\lambda\rho}$, where $S^{\lambda\rho}$ denotes an $S_{n-1}$ irrep of Young diagram $\rho$ (with $n-1$ boxes) contained in the $S_n$ irrep $S^\lambda$. The so-called \textit{branching rule} guarantees that the decomposition is multiplicity-free, i.e., each distinct $S_{n-1}$-irrep $S^{\lambda\rho}$ appear\risi{s} only once in the decomposition. The so-called Bratteli diagrams in Fig.\ref{Bratteli} show how how different irreps are decomposed. Continuing the decomposition process for $S_{n-2},...,S_1$, the original space $S^\lambda$ will be written as a direct sum of 1-dimensional subspaces ($S_1$-irreps are 1-dimensional):
\begin{align*}
	S^\lambda & = \bigoplus_{S_{n-1},\rho} S^{\lambda\rho} = \cdots = \bigoplus_{S_{n-1},\rho} \cdots \bigoplus_{S_1,\tau} S^{\lambda\rho\sigma \cdots \tau}.
\end{align*}
Each 1-dimensional subspace $S^{\lambda\rho\sigma \cdots \tau}$ can be represented by a nonzero vector in it. Normalizing them, we obtain an orthonormal basis $\{\ket{v_T}\}$ of $S^\lambda$ called the \textit{Gelfand--Tsetlin basis (GZ)} or \textit{Young-Yamanouchi basis}. Indices $\lambda, \rho, \sigma,...,\tau$ form a path in the Bratteli diagram (see Fig.\ref{Bratteli}) and can be used to define a standard Young tableau $T$ (Fig.\ref{Path}). Young basis vectors are in one-to-one correspondence with standard Young tableaux \cite{Sagan01}. The branching rule is also discussed in $\operatorname{SU}(d)$ representation theory and some authors refer to the $\operatorname{SU}(d)$-irrep basis as the GZ-basis if it is constructed in a similar manner. A more comprehensive discussion about building Schur basis by SU(d) Clebsch-Gordan decomposition and by $S_n$ branching rule is postponed in Section \ref{s4}.

\begin{figure}[H]
	\centering
	\includegraphics[width=3.5in]{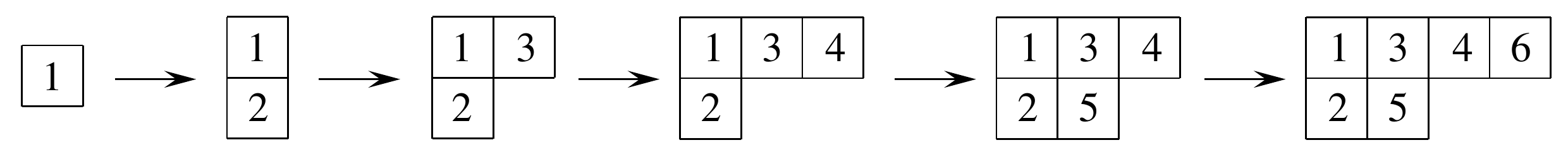}
	\caption{The Standard Young tableau defined by the path in Fig.\ref{Bratteli}}
	\label{Path}
\end{figure}

Let us introduce the central concept used in our work:

\begin{definition}
	For $1 < k \leq n$, the \textit{Young-Jucys-Murphy element}, or YJM-element for short, is defined as a sum of transpositions $X_k = (1,k) + (2,k) + \cdots + (k-1,k) \in \mathbb{C}[S_n]$. We set $X_1 = 0$ as a convention. 
\end{definition} 

As the name indicates, this concept was developed by Young \cite{Young1977}, Jucys \cite{Jucys1974} and Murphy \cite{Murphy1981}.  Okounkov and Vershik showed that  YJM-elements generate the \textit{Gelfand-Tsetlin subalgebra} $\operatorname{GZ}_n \subset \mathbb{C}[S_n]$~\cite{Okounkov1996}, which is a maximal commutative subalgebra consisting of all centers $Z[S_k]$ of $\mathbb{C}[S_k]$ for $k=1,...,n$. Another striking fact is that all YJM-elements are strictly diagonal (indeed, representation of $\operatorname{GZ}_n$ consists of all diagonal matrices) in a Young basis $\{\ket{v_T}\}$ whose eigenvalues can be read out directly from the standard Young tableaux $T$. To be precise, let $\lambda$ be a Young diagram. Since this is a 2-dimensional diagram, we can naturally assign integer coordinates to its boxes. The \textit{content} of each of its boxes is determined by the $x$-coordinate minus $y$-coordinate. Suppose $T$ is a standard Young tableau of $\lambda$. Arranging all contents with respect to $T$, we obtain the \textit{content vector} $\alpha_T$. For instance, the Young diagram $\lambda = (4,2)$ has contents $0,1,2,-1$. The specific standard tableau $T$ in Fig.\ref{Path} has content vector $\alpha_T = (0,-1,1,2,0,3)$. Let $\ket{v_T}$ denote a Young basis vector corresponding to $T$. Measuring by YJM-elements, we have $X_1 \ket{v_T} = 0,~ X_2 \ket{v_T} = -\ket{v_T},~ X_3 \ket{v_T} = \ket{v_T}$ and so forth. Each Young basis vector is determined uniquely by its content vector (see \cite{Okounkov1996} and \cite{Tolli2009} for more details).


\section{$S_n$ Convolutional Quantum Alternating Ans{\"a}tze}  \label{s3}

Consider the YJM-elements $\{ X_1, ..., X_n \}$ introduced in Section \ref{s2} which generate the maximally commuting subalgebra $\operatorname{GZ}_n$. The use of YJM-elements allows us to design the following \textit{mixer Hamiltonian}:
\begin{align}
	H_M = \sum_{i_1,..., i_N} \beta_{i_1... i_N}X_{i_1} \cdots X_{i_N}.
\end{align}
This YJM-Hamiltonian is still strictly diagonal under the Young basis. As there is an efficient quantum Schur transform $U_{\operatorname{Sch}}$ (QST) which transforms the computational basis to Schur basis even for qudits with gate complexity $\operatorname{Poly}(n, \log d,\log(1/\epsilon))$ \cite{harrow2005PhD,Harrow06,krovi}. It is reasonable to assume that we can initialize any Young basis element $\ket{\Psi_{\operatorname{init}}}$ from $S_n$-irreps via QST. Moreover, given a problem Hamiltonian $H_P$ with $\operatorname{SU}(d)$ symmetry, it can be written as $H_P = \wt{\pi} (f)$ for some $f \in \mathbb{C}[S_n]$ by the double commutant theorem and Wedderburn theorem (see Section \ref{s2} and SM). Inspired by QAOA (Pauli $X$ is diagonal wrt. $\ket{+}^{\otimes n}$) \cite{Farhi_2014}, we thus propose the following ansatz:
\begin{align*} 
	\cdots \exp (-i H_M ) \exp (-i \gamma H_P) \exp  (-i H_M )  \exp(-i \gamma^{\prime} H_P) \cdots.
\end{align*}
To summarize, we have defined a family of mixer operators $H_M$ parameterized by $\beta_{i_1... i_N}$ which is diagonal under the Young basis and naturally preserves each $S_n$-irrep determined by the initialized states. Following \cite{kondor1,kondor2} and one of our results \cite{Zheng2022PQC} which interpret quantum circuits as natural Fourier spaces, we call this family of ans{\"a}tze $S_n$-Convolutional Quantum Alternating Ans{\"a}tze ($S_n$-CQA).

Deviating from the original purpose of using QAOA to solve constraint satisfaction problems, our $S_n$-CQA are applied to problems with global $\operatorname{SU}(d)$ symmetry or permutation equivariance, as long as the problem Hamiltonian $H_P$ can be efficiently simulated in the circuit model. Indeed, most practical examples involve only 2- or 3-local spin-interactions such as the Heisenberg model studied in Section \ref{s5}, and thus by Theorem 2 of~\cite{Zheng2022PQC} can be efficiently simulated. On the other hand, the $k$-dependent constant in Corollary 1 from \cite{Zheng2022PQC} has exponential scaling. So in practice, we would like to have 3- or at most 4-local terms for the YJM-Hamiltonian simulation in ans{\"a}tze. We focus on the following mixer Hamiltonian consisting of only first and second order products of YJM-elements: 
\begin{align} \label{mixer2}
	H_M = \sum_{k \leq l} \beta_{kl} X_k X_l,
\end{align}
whose evolution can be efficiently simulated in $\mathcal{O}(n^4 \log(n^4/\epsilon)/\log\log(n^4/\epsilon))$. Below, we prove that with a mixture as in the Equation \eqref{mixer2}, the $S_n$-CQA ans{\"a}tze are able to approximate any unitary from every $S_n$-irrep block (Fig.\ref{SchurWeyl} $(b)$). This can be seen as a restricted version of universal quantum computation to $S_n$-irreps. Since the 4-local $S_n$-CQA is an all-you-need approximation algorithm within PQP+, it is strongly suggestive that the PQP+ class proposed in \cite{Zheng2022PQC} contains circuits that can approximate matrix elements of all the $S_n$ Fourier coefficients, for a polynomial number of alternating layers $p$. Moreover, the 4-local $S_n$-CQA is also the universal approximator for solutions of the problem with global $\operatorname{SU}(d)$ symmetry, such as the Heisenberg models, due to its nature as variational ans{\"a}tze.  


\subsection{Restricted Universality of $S_n$-CQA Ans{\"a}tze in $S_n$ Irreps} \label{s3.A}

We now present the main theoretical result of this paper: the $S_n$-CQA ans{\"a}tze approximate any unitary in any $S_n$-irrep decomposed from the system of qudits. This \emph{restricts} universal quantum computation on $\operatorname{U}(d^n)$ to $S_n$-irrep blocks (Fig.\ref{SchurWeyl} $(b)$) because our ans{\"a}tze preserve $\operatorname{SU}(d)$ symmetry. This is of interest for three reasons: $(a)$ the density result indicates that $S_n$-CQA ans{\"a}tze is a universal approximator in PQP+ proposed in \cite{Zheng2022PQC} and it is the theoretical guarantee of our numerical simulations. $(b)$ The result is valid for qudits under $\operatorname{SU}(d)-S_n$ duality and we show the advantage of working with $S_n$ as there is no need to deal with complicated $\operatorname{SU}(d)$ symmetry, generators in the proof. $(c)$ When changed from the Young basis to the computational basis, i.e., forgetting the  $\operatorname{SU}(d)$ symmetry, our results form a new proof to the universality of a broad class of QAOA ans{\"a}tze. $(d)$ It is shown in a recent work \cite{MarvianSUd} that $\operatorname{SU}(d)$-invariant/symmetric quantum circuits with $d \geq 3$ cannot be generated by 2-local $\operatorname{SU}(d)$-invariant unitaries. With the focus on locality, we verify that $S_n$-CQA ans{\"a}tze can be built by 4-local $\operatorname{SU}(d)$-invariant unitaries and 4-locality is enough to generate any $\operatorname{SU}(d)$-invariant quantum circuit up to phase factors.

Mathematically, we aim to show that the subgroup generated by $S_n$-CQA ans{\"a}tze is equal to the unitary group $\operatorname{U}(S^{\lambda})$ restricted to $S^{\lambda}$ decomposed from $V^{\otimes n}$. However, arguing directly on the level of the Lie group is complicated. Instead, we prove that the generated Lie algebra is isomorphic with the unitary algebra $\mathfrak{u}(S^{\lambda})$ restricted to $S^\lambda$. Then combining with some classical results from the theory of Lie group \cite{Kuranishi51,Yamabe1950} and the Okounkov-Vershik approach to $S_n$-representation theory \cite{Okounkov1996}, we complete the proof. We outline our results here and put all the proof details into the SM. 

Our first step is motivated by a classical result from the theory of Lie algebra: any semisimple Lie algebra can be generated by only two elements \cite{Kuranishi51}. Finding these elements would be tricky and encoding them by a quantum circuit would even be infeasible, so we will adopt a different routine and solve these problems gradually. We first work on the complex general linear algebra $\mathfrak{gl}(d,\mathbb{C})$ which is not semisimple, but facilitates our proof. To begin with, it is easy to find its Cartan subalgebra --- the collection $\mathfrak{d}(d)$ of all diagonal matrices. Let $M$ be a matrix with nonzero off-diagonal elements $c_{ij}$. It can be thought of as a perturbation from $\mathfrak{d}(d)$. We want to know how large the subalgebra generated by $M$ and $\mathfrak{d}(d)$ would be. More precisely:

\begin{lemma}\label{gl}
	Let $E_{ij} \in \mathfrak{gl}(d,\mathbb{C})$ be the matrix unit with entry $1$ at $(i,j)$ and $0$ elsewhere. Given any matrix $M$, let $\mathcal{I} \subset \{1,...,d\} \times \{1,...,d\}$ be the index set corresponding to nonzero off-diagonal entries $c_{ij}$ of $M$. Then the Lie subalgebra generated by $\mathfrak{d}(d)$ and $M$ contains 
	\begin{align*}
		\mathfrak{d}(d) \oplus \Big(  \bigoplus_{(i,j) \in \mathcal{I}} R_{ij} \Big),
	\end{align*}
	where $R_{ij}$ is the 1-dimensional root space spanned by $E_{ij}$.
\end{lemma}

Intuitively speaking, $\operatorname{GZ}_n$ defined in Section \ref{s2} corresponds precisely to the Cartan subalgebra of $\mathfrak{gl}(\dim S^\lambda,\mathbb{C})$. In proving Lemma \ref{gl}, we are required to used all basis elements of $\operatorname{GZ}_n$ rather than the $n$ YJM-generators \cite{Okounkov1996}. Thus we need to employ all high-order products $X_{i_1} \cdots X_{i_N}$ of YJM-elements (as $\dim S^\lambda$ increases exponentially for large number of qudits Fig.\ref{fig_scale}) and that pose the first problem for a practical ansatz design, which requires $k$-local $\mathbb{C}[S_n]$ Hamiltonian in order to be efficiently simulated by quantum circuits \cite{Zheng2022PQC}. This problem is solved in Lemma 2.6 the SM with the help of Okounkov-Vershik theorem \cite{Okounkov1996,Tolli2009}. We prove that the collection $\{X_i, X_k X_l\}$ of first- and second-order YJM-elements, while in general cannot form a basis for $\operatorname{GZ}_n$, are enough to establish Lemma \ref{gl}. As a reminder, merely taking the original YJM-elements $X_i$ is not sufficient and we provide counterexamples in the SM. 
As $X_i$ are 2-local, this result also provides some insights on the fact that 2-local $\operatorname{SU}(d)$-invariant unitaries cannot generate all quantum circuits with $\operatorname{SU}(d)$ symmetry in the general case \cite{Marin1,Marin2,MarvianSUd}.

As another ingredient of $S_n$-CQA ans{\"a}tze, the problem Hamiltonians $H_P$ of interest are complicated in general and hard to diagonalize classically. It also forms the other part (the matrix $M$) in generating the Lie algebra in Lemma \ref{gl}. For the purpose of easy implementation, we show in Lemma \ref{u} that $H_P$ only needs to be \emph{path-connected} or \emph{irreducible} in the language of graph theory. A Hamiltonian is of this kind if its associated index graph $\mathcal{G}_{H_P}$ is connected. For example, the Pauli $X$ and $Y$ are path-connected while $Z$ is not. We further prove in Lemma 3.1 in the SM that the 2-local Hamiltonian $H_S = \sum_{i = 1}^{n-1} (i,i+1)$ defined by all adjacent transpositions $(i,i+1) \in S_n$ is path-connected. We will discuss path-connectedness further after Theorem \ref{ex2} as well as in Section \ref{s5}. It is also seen in the famous Perron-Frobenius theorem and applied to graph theory. 

\begin{lemma}\label{u}
	Let $H_P$ be a path-connected Hamiltonian. Then the generated Lie algebra $\langle \mathfrak{d}(d), H_P \rangle = \mathfrak{gl}(d,\mathbb{C})$. Consider $\mathfrak{d}_\mathbb{R}(d)$ consisting of all real-valued diagonal matrices. Generated over $\mathbb{R}$, $\langle i\mathfrak{d}_\mathbb{R}(d), iH_P \rangle_\mathbb{R} = \mathfrak{u}(d)$. 
\end{lemma}

Since YJM-elements as well as their high order products have real diagonal entries under Young basis, we concretize $\mathfrak{d}(d)$ by $\{X_i, X_k X_l\}$. With all these preparations, we consider the subgroup $H$ defined in pure algebraic sense by alternating exponentials of $iX_k X_l$ and $iH_P$ where $H_P$ is path-connected. To verify that $H$ is a Lie group (with smooth structures \cite{Hilgert2012}), we apply another classical theorem due to Yamabe \cite{Yamabe1950,Goto1969} and conclude with:

\begin{theorem}\label{ex2}
	Restricted to any $S^\lambda$ with isomorphic copies decomposed from $V^{\otimes n}$, the subgroup generated by $X_k X_l$ with any path-connected Hamiltonian $H_P$ equals $\operatorname{U}(S^\lambda)$. Then a $S_n$-CQA ansatz is written as
	\begin{align} \label{CQA}
		\begin{aligned}
			\cdots & \exp(-i \sum_{k,l} \beta_{kl} X_k X_l ) \exp(-i\gamma H_P)  \\
			& \exp(-i \sum_{k,l} \beta'_{kl} X_k X_l) \exp(-i\gamma' H_P) \cdots,
		\end{aligned}
	\end{align}
	where we redefine $X_1$ as $I$ with which any first-order YJM-element $X_i$ can be written as $X_i X_1$.
\end{theorem}

Consider the case when $H_P$ is not path-connected. That is, $H_P$ is block diagonal (after a possible re-cording of basis elements) in $S^\lambda$. It is straightforward to check that Theorem \ref{ex2} still holds within each sub-block of $H_P$. Suppose our task is to find the lowest eigenstate $\ket{v_0}$ of $H_P$ within $S^\lambda$. There is generally no prior knowledge about which sub-block $v_0$ is in. The brute-force way to find the minimum is by taking a collection of initial states from each of these sub-blocks and applying the theorem repeatedly. One way to do this is by implementing the efficient QST which gives us access to all Young basis elements. The state initialization proposed in Section \ref{s4.B} with constant-depth may take a hit (forcing the depth of the circuit to increase) if the problem Hamiltonian is not path-connected. 


\subsection{Universality of QAQA}\label{s3.B}

The first proof of the universality of the QAOA ans{\"a}tze was given in \cite{Lloyd_2018}, where the authors considered problem Hamiltonian of the first-order and second-order nearest-neighbor interaction. \cite{Morales_2020} subsequently generalized the result to broader families of ans{\"a}tze defined by sets of graphs and hyper-graphs. We now describe a new proof based on the techniques developed in this paper that covers novel, broader family of QAOA ans{\"a}tze. More precisely, we change to the computational basis $\{\ket{e_i}\}_{i=1}^{2^n}$ in which all tensor products $\tilde{Z}_{r_1...r_s} \vcentcolon = Z_{r_1} \otimes \cdots \otimes Z_{r_s}$ of Pauli basis can span any diagonal matrix. In the language of Lie algebra, $Z_i$ generates the Cartan subalgebra $\mathfrak{d}(2^n)$ of $\mathfrak{gl}(2^n,\mathbb{C})$ (comparing with the case of $\operatorname{GZ}_n = \mathfrak{d}(S^\lambda)$ under Young basis). Let $H_X$ be the uniform summation of Pauli $X$ operators (we do not write its explicit form to avoid any confusion with the notation of YJM-elements). The Hamiltonian $H_X$ is path-connected under $\{\ket{e_i}\}$. To restrict the using of high order Pauli $Z$ operators analogously as we did for YJM-elements, we prove in Lemma 5.1 that the Hamiltonian $H_Z$ composed by $\{Z_i, Z_k Z_l\}$ are enough to establish Lemma \ref{gl} \& \ref{u} with $H_X$ in the present setting. Unlike the $H_Z$ used in \cite{Lloyd_2018} which contains only nearest neighbor terms $Z_j Z_{j+1}$, we take all second-order products $Z_k Z_l$ in our proof. The resulting Hamiltonian $H_Z$ is still simple though and the proof works for both odd and even number of qubits  \cite{Morales_2020}. Moreover, replacing $H_X$ by any other path-connected Hamiltonian, e.g., an unfrustrated Heisenberg Hamiltonian with boundary condition \cite{Tasaki2020}, still guarantees the universality, and this fact enables one to experiment with a wide range of mixer Hamiltonians. In summary,

\begin{theorem}\label{QAOA-Univ}
	Let $H_X$ be any path-connected Hamiltonian on computational basis, the group generated by the QAOA-ansatz with $H_X, H_Z = \sum \beta_{kl} Z_k Z_l$ equals $\text{U}(2^n)$, i.e., it is universal.  
\end{theorem}


\subsection{Four-Locality of $\operatorname{SU}(d)$-symmetric Quantum Circuits}\label{s3.C}

A well-known result in~\cite{Vlasov2001} states that any quantum circuit can be generated by 2-local unitaries for qubits as well as for qudits. It has been shown in a recent work \cite{MarvianSUd} that this statement fails to hold when we impose the $\operatorname{SU}(d)$ symmetry on qudits system with $d \geq 3$. Let $\mathcal{V}_k$ denote the subgroup generated by $k$-local $\operatorname{SU}(d)$-invariant unitaries, so $\mathcal{V}_2 \neq \mathcal{V}_n$, where $\mathcal{V}_n$ stands for all the irrep blocks from Fig. \ref{SchurWeyl} $(b)$. On the other hand, we use $\operatorname{U}(S^\lambda)$ in Theorem \ref{ex2} which specifies one (with equivalent copies) of them in searching ground state of Heisenberg Hamiltonian in Section \ref{s5}).   

Counting all inequivalent $S^\lambda$ is an interesting problem of its own, especially when studying the subgroups $\mathcal{V}_k \subset \operatorname{U}(d^n)$ induced by symmetry, but it would cause a \emph{phase factor problem}: one may not be able to manipulate relative phase factors of unitaries generated in inequivalent $S^\lambda$ arbitrarily. We could simply ignore these phase factors as they make no difference in measurements respecting the symmetry. Then we consider $\mathcal{SV}_k \subset \mathcal{V}_k$ restricted to $\operatorname{SU}(S^\lambda)$ for all $S^\lambda$ decomposed from $V^{\otimes n}$. It is shown in \cite{Marin1,Marin2,MarvianNature,MarvianSU2} that $\mathcal{SV}_2 = \mathcal{SV}_n$ when $d = 2$. However, \cite{Marin1,Marin2} prove in a pure math flavor by Brauer algebra from representation theory that the statement fails when $d \geq 3$ and \cite{MarvianSUd} shows by constructing an counterexample based on the qudit-fermion correspondence that $\mathcal{V}_2$ is not even a 2-design. That is, the distribution of unitaries generated by random 2-local unitaries cannot converge to the Haar measure of $\mathcal{V}_n$. With results about CQA developed above, we prove the following theorem:

\begin{theorem}\label{4-local}
	Ignoring phase factors, $\operatorname{SU}(d)$-invariant quantum circuits can be generated by 4-local $\operatorname{SU}(d)$-invariant unitaries for any $d \geq 2$. Using group-theoretical notation, $\mathcal{SV}_4 = \mathcal{SV}_n \subset \operatorname{CQA}$.
\end{theorem}

We sketch the proof strategy and leave the details into the SM: We define the subgroup generated by CQA, still denoted by CQA for simplicity, by a 2-local path-connected Hamiltonian $H_S = \sum_{i = 1}^{n-1} (i,i+1)$ mentioned in Section \ref{s3.A}. Since second order YJM-elements are at most $4$-local, one can intuitively conclude that $\mathcal{SV}_4 = \mathcal{SV}_n$. Moreover, in contrast to Theorem \ref{ex2} which addresses the universality restricted to one fixed $S^\lambda$, we now consider all inequivalent $S^\lambda$ from $V^{\otimes n}$ and our method only handles the problem when ignoring phase factors, Thus we claim that $\mathcal{SV}_4 = \mathcal{SV}_n$ and they are all included in CQA because CQA contains generators with nontrivial phases: e.g., $e^{i\theta I}$. We discuss more details about the phase factor problem by $S_n$ representation theory and show that $\operatorname{CQA}$ is a compact subgroup of $\mathcal{V}_4 \subsetneqq \mathcal{V}_n$ generally in the SM.


\section{Correspondence between Spin Labels and Content Vectors}\label{s4}

As introduced in Section \ref{s2}, Young basis vectors are labeled by content vectors via YJM-elements. A similar phenomenon is also seen for the $\operatorname{SU}(d)$ irrep basis vectors constructed by Clebsch-Gordan decompositions \cite{harrow2005PhD, pqc, sergii1} under which the space decomposes as Fig.\ref{SchurWeyl} $(a)$: they are labeled by $d-1$ Casimir operators \cite{Biedenharn1}. We now turn to the question whether these two labeling schemes are equivalent in a certain sense explained in the following. This was conjectured to be true in \cite{harrow2005PhD} and surfaced again in \cite{krovi} when the author introduced an efficient Quantum Schur Transform (QST). An affirmative answer to this conjecture is crucial in this work for three reasons: $(a)$ The Young basis is algebraic. Thus, the gate action drawing from the group algebra $\mathbb{C}[S_n]$ is basis-independent. In particular, it can be implemented directly in the computational basis without computing the Fourier coefficients -- this is a key observation that underpins the super-exponential quantum speed-up. $(b)$ This identification allows us to apply both classical tools from $\operatorname{SU}(d)$ representation theory as well as Okounkov-Vershik approach to Schur basis no matter how it is established. As an example, we show in Section \ref{s4.B} an efficient algorithm to generate Schur basis states required for optimization and learning tasks. $(c)$ A detailed examination on Schur basis enables us covert all the previous results about $S_n$-CQA to, what we call, $\operatorname{SU}(d)$-CQA with $S_n$ symmetry.  

For two-row Young diagrams, this conjecture was shown to be correct in \cite{Pauncz2018}, where the author studied the question by $\frac{1}{2}$-spin eigenfunctions instead of YJM-elements. The general case for $\operatorname{SU}(d)-S_n$ duality still holds and can be proven in a surprisingly easy way using YJM-elements and the Okounkov-Vershik approach. We present details in the SM.

\begin{lemma}\label{J_k^2-X_i}
	Under $\operatorname{SU}(d)-S_n$ duality, sequentially coupled Casimir operators commute with YJM-elements.
\end{lemma}

As a brief illustration of this result, let us consider the \textit{sequentially coupled total spin basis} $\ket{j_1,...,j_n; m}$ of $\operatorname{SU}(2)$. The \textit{spin component} $m$ and \textit{spin labels} $j_k$ are determined by spin operator $S_z^n$ as the summation of all half Pauli $Z$ matrices $\frac{1}{2}Z_i$ at each $i$ site and sequential coupled Casimir operators $J_k^2 = (S_x^k)^2 + (S_y^k)^2 + (S_z^k)^2$ respectively (we abuse our language for simplicity as true eigenvalues of $J_k^2$ are $j_k(j_k+1)$). Since they commute with YJM-elements, $J_k^2 X_i \ket{j_1,..., j_n, m } = X_i J_k^2 \ket{j_1,...,j_n; m}$. It is well-known from linear algebra that commutative operators can be simultaneously diagonalized and we elaborate this fact with the following theorem:

\begin{theorem}\label{Spin-Content}
	YJM-elements are strictly diagonal under the $\operatorname{SU}(d)$ irrep basis built by sequential Clebsch-Gordan decompositions. Conversely, sequentially coupled Casimir operators are strictly diagonal under the Young basis decomposed by branching rule. 
\end{theorem}

We now illustrate by examples the correspondence between spin labels and content vectors for the simplest $\operatorname{SU}(2)-S_n$ duality, then we go to general case. Let $\ket{j_1,...,j_n;m}$ be any $\operatorname{SU}(2)$ irrep basis vector. Theorem \ref{Spin-Content} says that it is also a Young basis element, thus we can talk about its eigenvalues (content vector) $(\alpha_T(1),...,\alpha_T(n))$ with respect to the YJM-elements (recall that $T$ denotes a denotes a standard Young tableau, or equivalently the corresponding GZ-path from the Bratteli diagram like that from Fig.\ref{Bratteli} \& \ref{Path}). An \emph{equivalence} between two labeling schemes means the spin label $J = \{j_1,...,j_n\}$ uniquely determines the content vector $\alpha_T = (\alpha_T(1),...,\alpha_T(n))$ and vice versa. 

\noindent \emph{$\operatorname{SU}(2)$ case:} For brevity, let us denote $\operatorname{SU}(2)$ irrep basis vectors by $\ket{J; m}$. It is possible to find two basis elements $\ket{J; m}, \ket{J'; m}$ with the same spin component $m$ but different spin labels. This is due to the fact that $(\mathbb{C}^2)^{\otimes n}$ would decompose into copies of isomorphic $\operatorname{SU}(2)$ irreps (Fig.\ref{SchurWeyl} $(a)$). On the other hand, a Young basis element is then denoted by $\ket{\alpha_T;\mu}$ where $\mu$, as explained in Section \ref{s2}, comes with choosing the permutation module $M^\mu$. It is also possible to find two basis elements $\ket{\alpha_T;\mu}, \ket{\alpha_T;\mu'}$ with the same content vector $\alpha_T$ but from different permutation modules $M^\mu$ (Fig.\ref{SchurWeyl} $(c)$). Let us forget the problem of copies or multiplicities for a while and only focus on the correspondence between $J$ and $\alpha_T$. We would come back to discus then in Section \ref{s4.A}. Let $\ket{\alpha_T;\mu}$ be a Young basis element such that $\alpha_T$ equals $(0,-1,1,2,0,3)$ in Fig.\ref{Path}. It is also an $\operatorname{SU}(2)$ irrep basis vector. Acted on by $J_1^2$, the first spin label is definitely $j_1 = \frac{1}{2}$. To measure the second spin label, let us apply Schur-Weyl duality to the subset system consisting of only the first two qubits. Since $\ket{\alpha_T;\mu}$ is constructed by branching rule, it can be seen as a Young basis element of $S_2$ irreps of Young diagram $\lambda = = (\lambda_1, \lambda_2) = (1,1)$ (read off from the first two elements of $\alpha_T$). Schur-Weyl duality says that $\ket{\alpha_T;\mu}$ should stay in the $\operatorname{SU}(2)$ irrep denoted by the same Young diagram, hence $j_2 = \frac{1}{2}(\lambda_1 - \lambda_2) = 0$. Inductively, $j_3 = \frac{1}{2}$ and we obtain $J = (\frac{1}{2}, 0, \frac{1}{2}, 1, \frac{1}{2}, 1)$. As a brief comment on multiplicities, since the total spin (last spin label) is $1$, there are three possible choices of $z$-spin components $m = 1,0,-1$. Correspondingly, $\mu$ from $\ket{\alpha_T;\mu}$ can be three different permutation modules (Fig.\ref{SchurWeyl} $(c)$). The mechanism for reading off content vectors from spin labels is simialr.

\noindent \emph{$\operatorname{SU}(3)$ case:} Note that the pattern of constructing $\operatorname{SU}(2)$ spin labels is simply the familiar branching rule seen in $\operatorname{SU}(2)$-irreps \cite{sergii1, sergii2}. We now instantiate with $d = 3$ to show how this pattern generalizes. Let $\boldsymbol{0},\boldsymbol{3},\overline{\boldsymbol{3}},\boldsymbol{8}$ denote the trivial, the fundamental, the conjugate and the adjoint representations of $\operatorname{SU}(3)$. Then we consider the following coupling scheme:
\begin{align*}
	\boldsymbol{3} \otimes \boldsymbol{3} = \overline{\boldsymbol{3}} \oplus \boldsymbol{6}  \, \Rightarrow \,
	\overline{\boldsymbol{3}} \otimes \boldsymbol{3} = \boldsymbol{0} \otimes \boldsymbol{8} \, \Rightarrow \,
	\boldsymbol{0} \otimes \boldsymbol{3} = \boldsymbol{3},
\end{align*}
where we coupled 4 qudits in which we take the GZ-path corresponding to $\alpha_T = (0,-1,-2,1)$ and ended up with the Young diagram $\lambda = (2,1,1)$. 

The group $\operatorname{SU}(3)$ has two Casimir operators 
\begin{align}\label{Casimir}
	C_1 = \sum_{i = 1}^8 T_i^2, \quad C_2 = \sum_{i,j,k}^8 d_{ijk} T_i T_j T_k,
\end{align}
where $T_i = \frac{1}{2}\lambda_i$ are half of the Gell-Mann matrices and $d_{ijk}$ are determined by the anti-commutation relation $\{T_i, T_j\} = \frac{1}{3} \delta_{ij} + d_{ijk} T_k$. Thus the $k$th sequential coupling of these operators, denoted by $(C_1, C_2)_k$, corresponds to the YJM-element $X_k$ and they are used to record irreps like $\boldsymbol{3}, \overline{\boldsymbol{3}}, \boldsymbol{0}, \boldsymbol{3}$ appearing in the above example and yield ``spin labels" $(1,0), (0,1), (0,0), (1,0)$ which are highest weights for $\operatorname{SU}(3)$ irreps. For a general $\operatorname{SU}(d)-S_n$ duality, each YJM-element $X_k$ corresponds to a pair of $d - 1$ Casimirs \cite{Biedenharn1} for the $k$th sequential coupling. It is therefore more concise to use YJM-elements in general case as sequential coupling and branching rule decomposition are equivalent in describing Schur basis.


\subsection{More facts about State Labeling and CQA with $S_n$ symmetry}\label{s4.A}

With spin label-content vector correspondence, we denote a Schur basis vector by $\ket{\alpha_T;\mu_S}$, where $\alpha_T$ is its content vector. In the $S_n$ picture, $\alpha_T$ tells us exactly the path to restrict an $S_n$ irrep determined by the Young diagram $T$ to $S_{n-1}$ irrep and so forth. However, there are copies of that $S_n$ irrep decomposed from the entire Hilbert space and $\mu_S$ labels the multiplicity. One may wish to distinguish these isomorphic copies by permutation modules like Fig.\ref{SchurWeyl} $(c)$. However, as mentioned in Section \ref{s2}, when $d \geq 3$ isomorphic copies of $S_n$ irreps can even be found from the same permutation module $M^\mu$ \cite{Sagan01,Tolli2009} and hence the superscript $\mu$ is no longer enough to identify $\mu_S$. 

Interestingly, this problem can be solved in the $\operatorname{SU}(d)$ picture, in which $\alpha_T$ tells us exactly the path to couple $\operatorname{SU}(d)$ irreps sequentially. Our finial destination $W_\lambda$ is uniquely determined by the path, but we now need to \emph{label} $\ket{\alpha_T;\mu_S}$ \emph{as a state} in $W_\lambda$. When $d = 2$, $\mu_S$ is simply taken as the spin-$z$ component $m$. When $d \geq 3$, we uses classical results by Gelfand and Tsetlin \cite{Gelfand2015} and Biedenharn \cite{Biedenharn2}. We illustrate the main idea for $d = 3$: consider weight diagrams of $\operatorname{SU(3)}$ irreps on Fig~\ref{WeightDiagram}. Each dot in a diagram stands for a basis vector of the irrep \cite{Greiner1994,Goodman2009}.
\begin{figure}[H]
	\centering
	\includegraphics[width=2in]{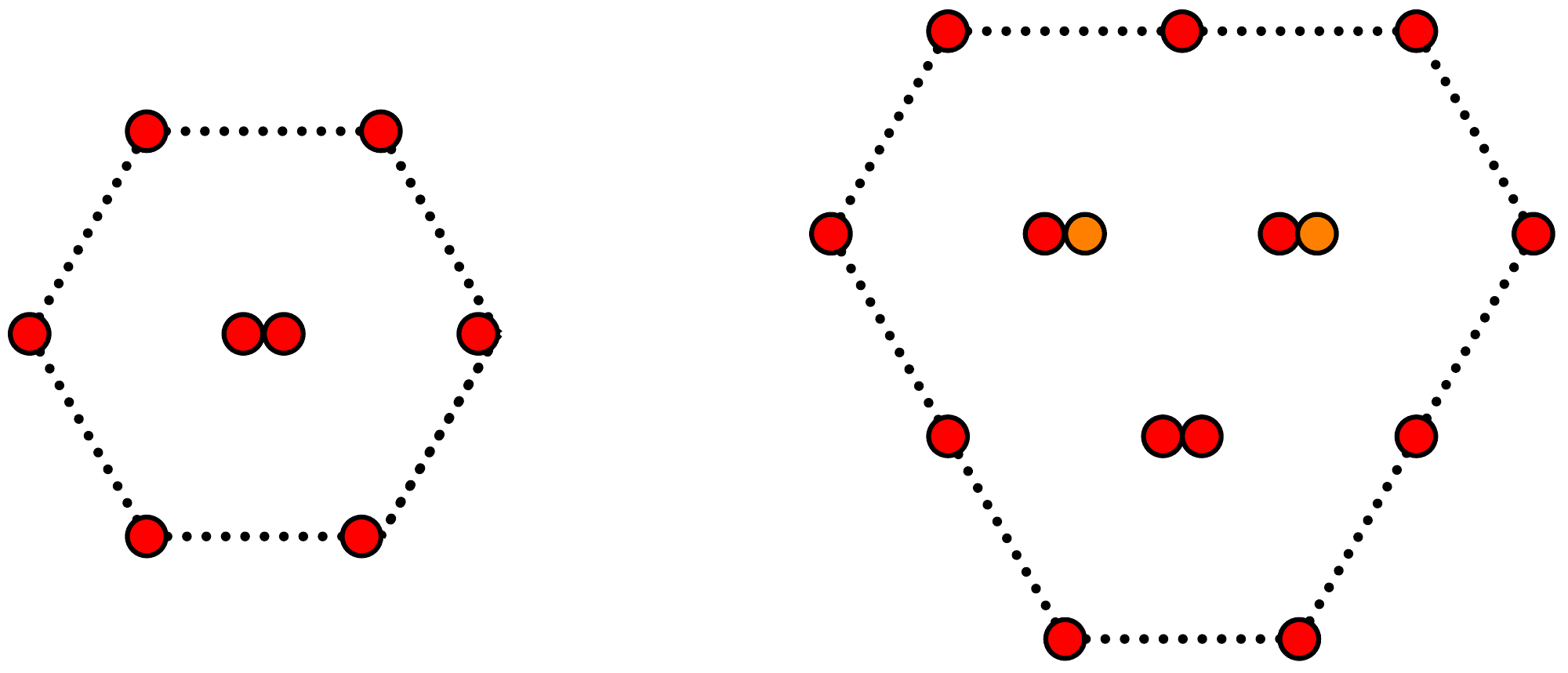}
	\caption{Typical weight diagrams of $\operatorname{SU(3)}$ irreps. Some weight vectors (dots) may occupy the same positions.}
	\label{WeightDiagram}
\end{figure}
With $\alpha_T$ being determined by YJM-elements, we only need to identify which dot $\ket{\alpha_T;\mu_S}$ corresponds to in the weight diagram of $W_\lambda$. Diagrammatically, these dots have planar coordinates which are rigorously called \emph{weights} measured by the isospin $I_3$ and hypercharge $Y$ operator of $\operatorname{SU}(3)$ \cite{Greiner1994,Goodman2009} just like measuring spin components by $S_z$ of $\operatorname{SU}(2)$. However, being different from the $\operatorname{SU}(2)$ case, some weight vectors (dots) occupy the same positions. It is known as the branching rule for $\operatorname{SU}(3)$ that dots with the same horizontal coordinate form irreps of $\operatorname{SU}(2) \subset \operatorname{SU}(3)$. For instance, two brown dots in Fig.\ref{WeightDiagram} form a spin-$1/2$ irrep while the other four red dots from the same horizontal line form a spin-$3/2$ irrep. Thus after measuring weights/positions by $I_3, Y$, we simply need to apply the $\operatorname{SU}(2)$ Casimir operator $J^2$ to discern dots occupying the same position.

In~\cite{Biedenharn2}, authors provide the recipe to find these operators for a general $\operatorname{SU}(d)$ group. Roughly speaking, we first employ $d - 1$ operators $\mu_i$, which span its Cartan subalgebra, to label the weights of a given basis vector. Then we take Casimir operators of $\operatorname{SU}(2) \subset \cdots \subset \operatorname{SU}(d-1)$ to distinguish  dots which occupy the same position. Their eigenvalues, which we record as $\mu_S$, form the so-called \emph{Gelfand-Tsetlin pattern} \cite{Gelfand2015} which is widely used to study $\operatorname{SU}(d)$ irreps. To construct a CQA model with $S_n$ symmetry, we replace the YJM-elements labeling any $S^\lambda$ basis states by the the following ones labeling any $W_\lambda$ basis states:
\begin{align*}
	\mu_1,...,\mu_{d-1}, C_1^{\operatorname{SU}(d-1)},...,C_{d-2}^{\operatorname{SU}(d-1)},...,C_1^{\operatorname{SU}(2)} = J^2.
\end{align*}
The number of required operators to build an $\operatorname{SU}(d)$-CQA ansatz is $\frac{1}{2}d(d - 1)$ -- a constant for fixed $d$, no matter how many qudits the system contains. However, Casimir operators of $\operatorname{SU}(d-1)$ are supported on $d-1$ qudits (see Eq.\eqref{Casimir} and \cite{Greiner1994}). Applying the same proof method from Theorem \ref{ex2}, $\operatorname{SU}(d)$-CQA is thus made of $2(d-1)$-local $S_n$-symmetric unitaries. As a simple example, $\operatorname{SU}(2)$ irreps basis states are uniquely determined by the summation $S_z^n$ of spin operator $\frac{1}{2}Z_i$ on each site $i$. Without having to employ Casimir operator, $S_z^n, (S_z^n)^2$ and a problem Hamiltonian $H_P$ already form an $\operatorname{SU}(2)$-CQA model. More precisely,

\begin{corollary}
	Let $H_P$ be any path-connected Hamiltonian on the $\operatorname{SU}(2)$ irreps basis basis, the CQA-ansatz generated by $H_P, S_z^n, (S_z^n)^2$ is dense in each $\operatorname{SU}(2)$-irrep block $\operatorname{U}(\dim W_\lambda)$.
\end{corollary}


\subsection{State Preparation for $S_n$-CQA Ans{\"a}tze}\label{s4.B}

To investigate evaluation of the matrix elements of $S_n$ Fourier coefficients, we were confined to the Young basis, which requires the implementation of Quantum Schur Transform \cite{harrow2005PhD,Harrow06,krovi}. However, for a wide variety of quantum machine learning and optimization tasks, such as determining the ground state sign structure of frustrated magnets, it is often advantageous to relax the constraints and ask how easy it is to initialize the states that live in any given $S_n$-irrep. An algorithm to initialize a state in the $S_n$-irrep with Young diagram being $(\frac{n}{2}, \frac{n}{2})$, is given in~\cite{Seki_2020}. We generalize this result to an arbitrary $S_n$-irrep in general $\operatorname{SU}(d)-S_n$ duality. The key is to utilize different permutation modules and multiplicities of $S_n$-irreps as in Fig.\ref{SchurWeyl} $(c)$. Simialr to the previous subsection, we construct the algorithm inductively: we first consider the $\operatorname{SU}(2)-S_n$ duality in which a $(\lambda_1, \lambda_2)$-$S_n$-irrep is dual to a spin-$(\lambda_1 - \lambda_2)/2$ irrep. Let
\begin{align}\label{psi-init}
	\ket{\Psi_{\operatorname{init}}} = \ket{\underbrace{0\cdots0}_{k \operatorname{many}}} \otimes \underbrace{ \ket{s} \otimes \cdots \otimes \ket{s}}_{\frac{n-k}{2} \operatorname{many}},
\end{align}
where $\ket{s} = \frac{1}{\sqrt{2}}(\ket{01} - \ket{10})$ is one of the Bell states and we assume $n-k$ is even. Then we have:

\begin{lemma}\label{YoungCombination}
	Let $\mu = \lambda = (\frac{n+k}{2}, \frac{n-k}{2})$. The initialized state $\ket{\Psi_{\operatorname{init}}}$ is contained in $S^\lambda$ and belongs to the permutation module $M^\mu$.
	\begin{proof}
		Acting by the spin operator $S^z = \sum_i S^z_i$, it is easy to check that the spin component of $\ket{\Psi_{\operatorname{init}}}$ equals $j = k/2$ hence it belongs to $M^\mu$. By Theorem \ref{Spin-Content} and discussion from the previous subsection, we have the expansion $\ket{\Psi_{\operatorname{init}}} = \sum_T c_T \ket{\alpha_T; k/2}$. Since by definition $J_+ \ket{\Psi_{\operatorname{init}}} = 0$, the Young diagram underlying each $\alpha_T$ from the summation must be the same and equals $\lambda$.
	\end{proof}
\end{lemma}

We now illustrate by several examples how to expand $\ket{\Psi_{\operatorname{init}}}$ as a linear combination of Young basis elements: $(a)$ let $\ket{\Psi_{\operatorname{init}}} = \ket{s}^{\otimes \frac{n}{2}}$. This is the state used in \cite{Seki_2020}. One can check by YJM-elements that it is a single Young basis element. $(b)$ for a more involved case, consider the $(4,2)$-irrep of $S_6$ and write: 
\begin{align*}
	& \ket{\Psi_{\operatorname{init}}} = \ket{00} \otimes \ket{s} \otimes \ket{s} \\
	= & \frac{2}{3}\ket{\alpha_{T_1}; 1} - \frac{\sqrt{2}}{3}\ket{\alpha_{T_2}; 1} - \frac{\sqrt{2}}{3}\ket{\alpha_{T_3}; 1}  + \frac{1}{3} \ket{\alpha_{T_4}; 1},
\end{align*}
where $\alpha_{T_i}$ corresponds to GZ-paths in Fig. \ref{Initialization}.

\begin{figure}[H]
	\centering
	\includegraphics[width=3.5in]{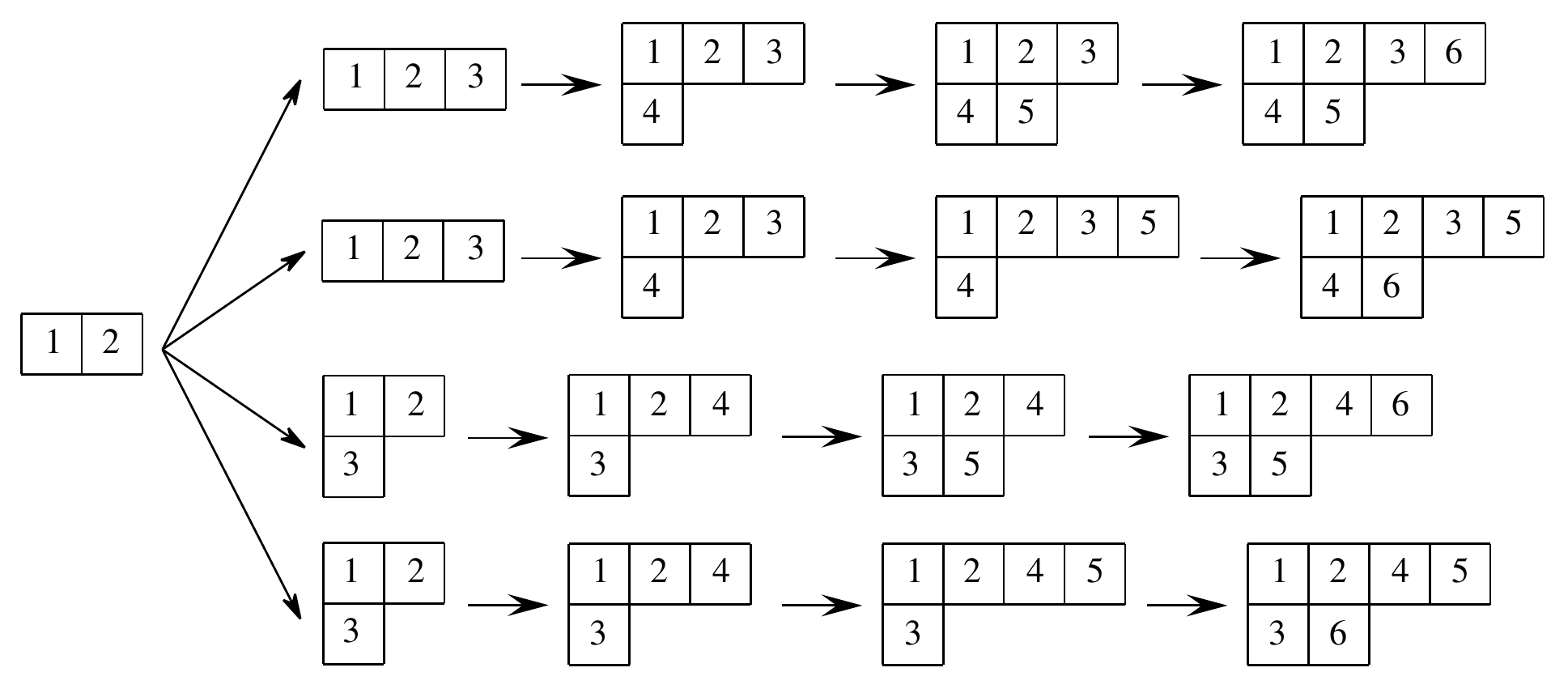}
	\caption{Decomposing the initial state by content vectors/spin labels.}
	\label{Initialization}
\end{figure}

The first two boxes in Fig.\ref{Initialization} corresponds to trivial irrep of $S_2$ acting on the subsystem formed by the first two qubits. Indeed, $S_2$ acts trivially on $\ket{00}$ from $\ket{\Psi_{\operatorname{init}}}$. As it would be more apparent to see how to get $\alpha_{T_i}$ in $\operatorname{SU}(2)$ picture, we use the spin label-content vector duality and trace the path of spin coupling. As the total spin of $\ket{00}$ and $\ket{00} \otimes \ket{s}$ are identical (Lemma \ref{YoungCombination}), there are two ways to add two more boxes: putting the third box on the RHS of the first two and then putting the forth on the bottom or conversely. Tensoring again with the singlet $\ket{s}$, we retrieve four branching paths in total. Moreover, by the same reason, it is easy to see that re-ordering tensor products of $\ket{0\cdots0}$ and $\ket{s}$ in Eq.\eqref{psi-init} yields a different expansion of Young basis elements for the same $S_n$ irreps. Fig.\ref{Initialization2} illustrates two more cases: $\ket{s} \otimes \ket{0} \otimes \ket{s} \otimes \ket{0}$ and $\ket{s} \otimes \ket{00} \otimes \ket{s}$. 

\begin{figure}[H]
	\centering
	\includegraphics[width=3in]{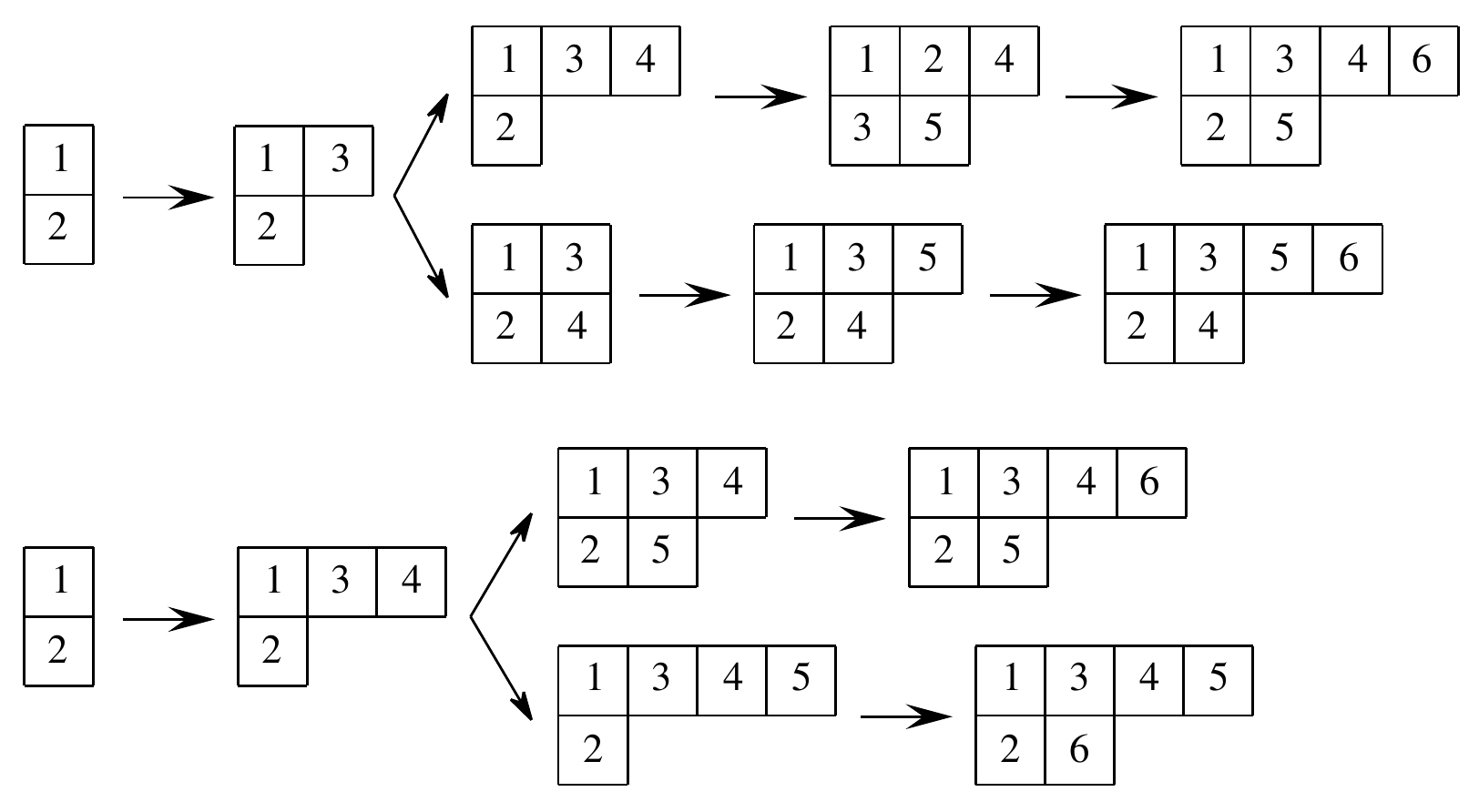}
	\caption{Reordering tensor products yields different Young basis expansions.}
	\label{Initialization2}
\end{figure}

This method can be generalized to $\operatorname{SU}(d)-S_n$ duality. For instance, when $d = 3$ to initialize states for three-row Young diagrams, let us consider the upper, down, strange states $u,d,s$ of $\boldsymbol{3}$. Let
\begin{align*}
	\ket{\Psi_0} = \ket{\underbrace{u \cdots u}_{k \operatorname{many}}} \otimes \underbrace{ (ud - du) \otimes \cdots \otimes (ud - du) }_{\frac{n-k}{2} \operatorname{many}}.
\end{align*}
This state lies in the $(n+k, n-k, 0)$-irrep. Tensoring with $SU(3)$-singlet $\ket{s} = (uds - usd + dsu - dus + sud - sdu)/\sqrt{6}$, $\ket{\Psi_{\operatorname{init}}} = \ket{\Psi_0} \otimes \ket{s}^{\otimes l}$ is a state from the $(n+k+l, n-k+l, l)$-irrep. Its expansion can still be tracked by the branching rule as in Fig. \ref{Initialization} \& \ref{Initialization2}. An $S_6$-CQA quantum circuit with state initialization described above can be seen in Fig.\ref{Circuit}.

\begin{figure*}[ht]
	\centering
	\includegraphics[width=0.9\textwidth]{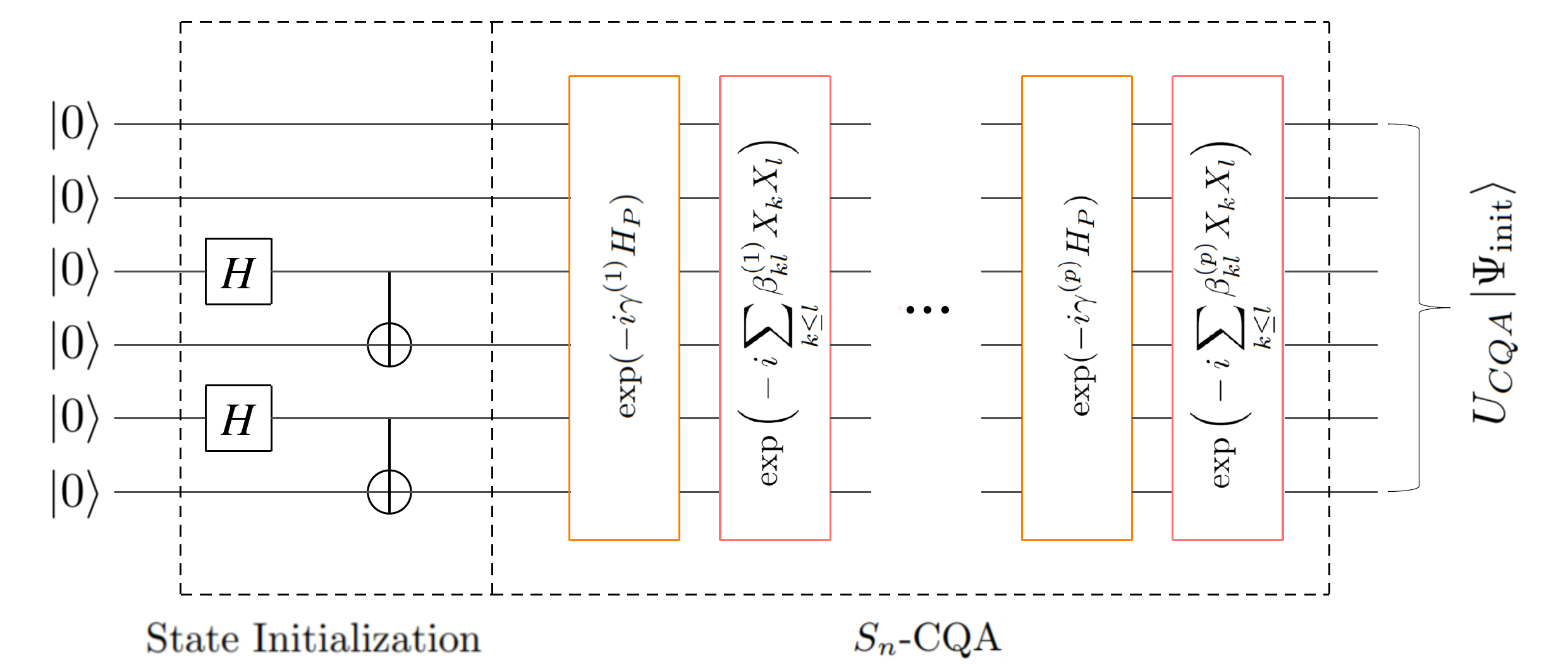}
	\caption{$S_6$-CQA circuit with state initialization in Fig.\ref{Initialization}.}
	\label{Circuit}
\end{figure*}


\subsection{Quantum Super-Polynomial Speedup}

For variational algorithms, typically one would make many measurements with updated parameters $\{\theta_\mu\}$ by some classical gradient descent scheme: 
\begin{align}
	\theta_\mu(t+1)=\theta_\mu(t)-\sum_\nu \eta_\mu(t) A_{\mu \nu}^{-1}(\theta(t)) \frac{\partial}{\partial \theta_\nu}\langle H\rangle_{\theta(t)},
\end{align}
where the learning rate tensor $A_{\mu \nu}(\theta(t))$ is often taken as identity matrix while $\eta_\mu = \eta$ is the learning rate. The quantity $\langle H \rangle_{\theta(t)}$ is the expectation value to minimize and $\frac{\partial}{\partial \theta_\nu}\langle H\rangle_{\theta(t)}$ is the derivative with respect to $\theta_{\nu}$. With an explicitly parameterized unitaries such as in our case, we can utilize the quantum circuits to measure the gradient of the expectation. 

Here, we refer taking one measurement at time $t$ as a \emph{query} and \emph{query complexity} as the total time $T$ in order to converge. The query complexity can be analyzed by recent development of the \emph{quantum neural tangent kernel} \cite{liu2022analytic}. Though it would be an interesting case to consider the bound on the query complexity to converge, in this work we only focus on showing that the circuit complexity \emph{per query} can be efficiently simulated on quantum circuits while this is not known in classical regime. 

\begin{theorem}
	Let $U^{(p)}_{\operatorname{CQA}}$ denote the CQA ans\"atze with $p$ alternating layers and let $H \in \mathbb{C}[S_n]$ be a SU($d$)-symmetric $k$-local Hamiltonian with most $N$ terms. Then for any $S_n$ irrep $S^\lambda$, 
	the Fourier coefficients: 
	\begin{align}
		\frac{\partial}{\partial \theta_{\mu}} \bra{\alpha_{T'}, \mu_S} U^{(p)\dagger}_{\operatorname{CQA}}(\boldsymbol{\theta}) H U^{(p)}_{\operatorname{CQA}}(\boldsymbol{\theta}) \ket{\alpha_T, \mu_S},
	\end{align}
	where $T,T'$ are standard tableaux of $\lambda$ (Fig,\ref{Bratteli} \& \ref{Path}) and $\mu_S$ records the multiplicity of $S^\lambda$ (Section \ref{s4.A}), can be simulated in $O(pN(\theta n^4 + k^2))$ with $\theta$ being the largest absolute values of parameters.
\end{theorem}

The proof is also put in the SM. Precisely, we assume that there exists an efficient Schur transform (QST) \cite{harrow2005PhD, krovi, kirby} with a polynomial overhead to prepare $\ket{\alpha_T,\mu_S}$. Calculating the Fourier coefficients over $S_n$ is a classically difficult question and the best classical algorithms $S_n$-FFT requires a factorial complexity \cite{Clausen_1993,Maslen_1998} as $S_n$ has $n!$ group elements and so is the dimension of its regular representation (see Wedderburn Theorem in the SM for more details). Therefore, comparing with the complexity of $S_n$-FFT, there is an super-exponential quantum speed per query. However as a caveat, the entire Hilbert space of $n$-qudits only scales exponentially with $n$ and $S_n$ irreps decomposed from the system by Schur-Weyl duality also scales exponentially. Therefore, it would be more reasonable and cautious to refers to a \emph{super-polynomial} quantum speed-up for $S_n$-CQA. 

Except comparing with $S_n$-FFT, recent work from \cite{tang2021quantum, tang2019quantum} proposes the notion of dequantization to compare the efficiencies of classical and quantum algorithms. Roughly speaking, with well-prepared quantum initial states, quantum algorithms can always be exponentially faster than the best counterpart classical algorithms. Assume classical algorithms also have efficient access to input. If the output can now be evaluated with at most polynomially larger query complexity then the quantum analogy, it is said to be dequantized with no genuine quantum speed-up. In our case, let us assume our initial states -- Schur basis elements $\ket{\alpha_T, \mu_S}$ or their linear combinations can be efficiently accessed with classical methods. Even though, dequantization still unlikely happens. Except conducting $S_n$-FFT, matrix representations of $\sigma \in S_n$ can also be efficiently sampled \cite{sergii1,sergii2}, but the method works exclusively for a single group element. To sample $U^{(p)}_{\operatorname{CQA}}(\boldsymbol{\theta}) \ket{\alpha_T, \mu_S}$ processed after $S_n$-CQA from Eq.\eqref{CQA}, the time evolution of CQA Hamiltonians is expand and approximated by at least super-polynomially many $S_n$ group elements (see the SM for more details) and hence is still thought to be classically intractable. 


\section{$\mathbb{C}[S_n]$ Symmetries of $J_1$-$J_2$ Heisenberg Hamiltonian } \label{s5}

The spin-$1/2$ $J_1$-$J_2$ Heisenberg model is defined by the Hamiltonian: 
\begin{align} \label{S_ham}
	\hat{H}_p=J_{1} \sum_{\langle i j\rangle} \hat{\boldsymbol{S}}_{i} \cdot \hat{\boldsymbol{S}}_{j}+J_{2} \sum_{\langle \langle i j\rangle \rangle } \hat{\boldsymbol{S}}_{i} \cdot \hat{\boldsymbol{S}}_{j},
\end{align}
where $\hat{\boldsymbol{S}}_{i} = (\hat{S}_{i}^{x}, \hat{S}_{i}^{y}, \hat{S}_{i}^{z})$ represents the spin operators at site $i$ of the concerned lattice. The symbols $\langle\cdots\rangle$ and $\langle\langle\cdots\rangle\rangle$ indicate pairs of nearest and next-nearest neighbor sites, respectively. The $J_1-J_2$ model has been the subject of intense research over its speculated novel spin-liquid phases at frustrated region \cite{Balents_2010}. The unfrustrated regime $(J_2 = 0$ or $J_1 = 0)$ for the anti-ferromagnetic Heisenberg model is characterized by the bipartite lattices, for which the sign structures of the respective ground states are analytically given by the \textit{Marshall-Lieb-Mattis theorem} \cite{Marshall1955,Lieb62,Tasaki2020}. As an important result, ground states of unfrustrated bipartite models are proven to live in the $S_n$ irrep corresponding to the Young diagram $\lambda = (n/2,n/2)$. By Schur-Weyl duality, this subspace is often referred as the direct sum of $\operatorname{SU}(2)$ invariant subspaces with total spin $J = 0$ in the context of physics (cf. Fig.\ref{SchurWeyl} $(a), (b)$). With this fact, algorithms like \cite{VieijraRBM} has been designed to enforce $\operatorname{SU}(2)$ symmetry at $J = 0$ and solve Heisenberg models without frustration.

The system is known to be highly frustrated when $J_1$ and $J_2$ are comparable $J_2/J_1 \approx 0.5$ \cite{Bukov_2021} and near the region of two phase transitions from Neel ordering to the quantum paramagnetic phase and from quantum paramagnetic to colinear phase, where no exact solution is known. Moreover, little is known about the intermediate quantum paramagnetic phase -- recent evidence of deconfined quantum criticality \cite{Nahum_2015,Wen_2016} sparked further interest in studying these regimes. Gaining physical insights in the intermediate quantum paramagnetic phase requires solving the problem of the ground state sign structure the system approaches the phase transition. Recently, there were a number of numerical attempts to address the existence of the $\operatorname{U}(1)$ gapless spin liquid phase, using recently the tensor networks \cite{Liu_2021}, restricted Boltzmann Machine (RBM) \cite{Nomura_2021}, convolutional neural network (CNN) \cite{Choo_2019,Liang_2018,Castelnovo_2020}, and graphical neural network (GNN) \cite{Kochkov_2021} -- all yielding partial progress. As a significant difference from the unfrustrated case, Marshall-Lieb-Mattis theorem does not hold generally and there is no guarantee that the ground state still lives at $J = 0$ or equivalently $\lambda = (n/2,n/2)$, which urges us to preserve the global SU(2) symmetry, which further gives us access to search in all inequivalent $S_n$ irreps decomposed from the system by Schur-Weyl duality.  


\subsection*{Global $\operatorname{SU}(2)$ Symmetry and Challenges in NQS Ans\"{a}tze}

Taking advantage of the global $\operatorname{SU}(2)$ symmetry, we address this problem in a different way: we recast the Hamiltonian in Equation \eqref{S_ham} by the following identity
\begin{align} \label{trans}
	\pi((i \: j)) = 2\hat{\boldsymbol{S}}_{i} \cdot \hat{\boldsymbol{S}}_{j} + \frac{1}{2}I,
\end{align}
with $\hat{\boldsymbol{S}}_{i}$ being further expanded as the half of standard Pauli operators $\{X, Y, Z\}$. Eq.\eqref{trans} was first discovered by Heisenberg himself \cite{Heisenberg1928, Klein_1992} (an elementary proof can be found in SM) and more recently noted by \cite{Seki_2020} in analyzing the ground state property of 1-D Heisenberg chain. As designed by products of exponentials of SWAPs (eSWAPs), the method proposed in \cite{Seki_2020} truly preserves the global SU(2) symmetry. As a brief comparison with $S_n$-CQA, eSWAP ans\"atze are universal in relevant sectors given by the SU(2) symmetry. However, this property no longer holds for qudits with $d \geq 3$ (Section \ref{s3.C}). As eSWAPs are non-commutative operators in general, there are various ways to place them in a quantum circuit. A more suitable perspective to describe the ansatz might be sampling them as 2-local SU(2) random circuits \cite{MarvianSUd}. On the other hand, the $S_n$-CQA ansatz is designed by alternating exponentials of the problem and mixer Hamiltonians $H_P, H_M$ just like the framework of QAOA. Similar to QAOA, $S_n$-CQA at large p corresponds to a form of adiabatic evolution with global SU($d$) symmetry, which could hint a theoretically guaranteed performance as $p$ is large (see Section VII in the SM).

An immediate consequence of using Eq.\eqref{trans} is that the resulting Heisenberg Hamiltonian can be expressed in the Young basis where every $S_n$-irrep is indexed by the total spin-label $j$. Mapping to this basis can be done using the constant-depth circuit state initialization in Section \ref{s4.B}. Using our $\mathbb{C}[S_n]$ variational ansatz leads to a more efficient algorithm by polynomially reducing the space. In the NISQ application, especially between 10 to 50 qubits, we have much better scaling see Fig.\ref{fig_scale}.

\begin{figure} [!ht] 
	\includegraphics[width = 0.40\textwidth]{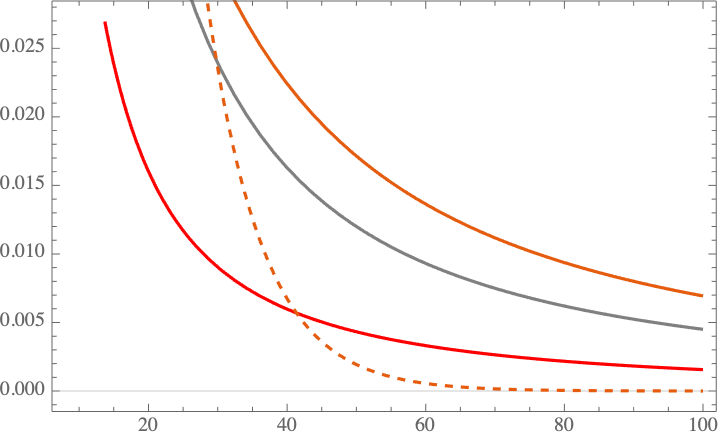} 
	\centering
	\caption{The scaling properties of small total spin irreps dimension respectively. The graph shows the scale: $2^n/ \dim S^{\lambda})$ with the partitions $\lambda_1 = (n/2, n/2) $ (red), $\lambda_2 = (n/2 +1, n/2-1)$ (grey), $\lambda_3 = (n/2 + 2, n/2 -2)$ (orange) The orange dashed line is $\exp(-n/8)$ the exponential decay. since the plot starts at $n=8$.} \label{fig_scale}
\end{figure}

Numerous efforts in applying NQS variational architecture to represent the complicated sign structure in the frustrated regime essentially use the energy as the only criterion for assessing its accuracy. This would result in the optimized low-energy variational states in frustrated regime still obeying the Marshall sign rules even though the true ground state is likely to deviate from it significantly \cite{Castelnovo_2020}, or breaks the $\operatorname{SU}(2)$ symmetry \cite{Choo_2019}.  The preservation of spatial symmetry has been the core topic of discussion in the literature, with proposed $\mathcal{C}_4$ equivariant CNN. However, on the 2D model Heisenberg model, the spatial symmetry consideration can only reduce the search space redundancy by a constant factor, thus scaling very poorly at even intermediate $n$. By reinforcing $\operatorname{SU}(2)$ symmetry, we achieve a polynomial reduction of Hilbert space and ensure the result to be physically reasonable, hence offering a second criterion to assess the variational ans\"{a}tze. 

\begin{figure*}[ht]
	\centering
	\includegraphics[width=0.9\textwidth]{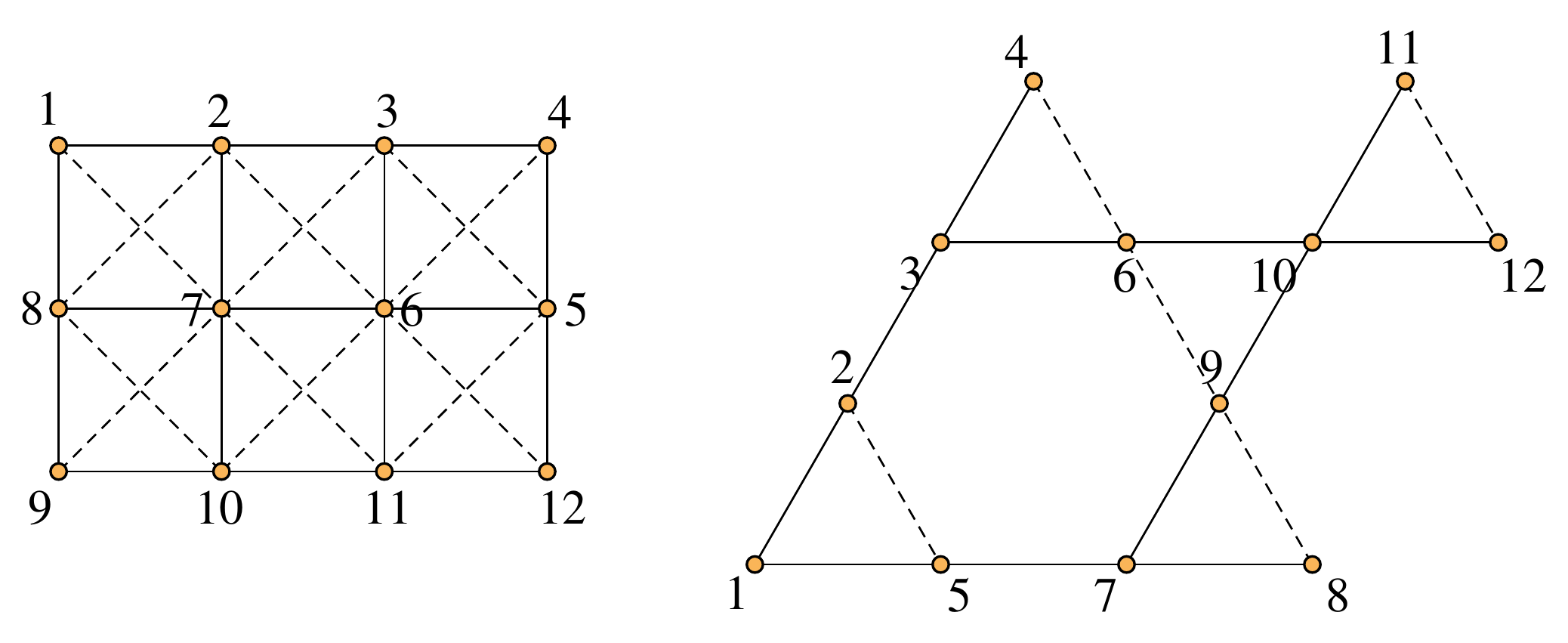}
	\caption{Square and Kagome lattices.}
	\label{Lattice}
\end{figure*}

The number of qubits scaling linearly with the number of qubits naturally circumvent the issue of having generalization property, a crucial property for the NQS ans{\"a}tze to function \cite{Westerhout_2020}. In fact, in a related work of us \cite{Zheng2022PQC}, we showed that, making use of the representation theory of the symmetric group, this leads to the super-exponential quantum speed-up. To this end, it is unlikely that any classically trained ans{\"a}tze are capable to reinforce the global SU(2) symmetry of the system. 


\subsection{Numerical Simulation}~\label{numsim}

We provide numerical simulations to showcase  the effectiveness of the $S_n$-CQA ans{\"a}tze, using JAX automatic differentiation framework \cite{JAX}.  The implementation of the $S_n$-CQA ans{\"a}tze utilizes the classical Fourier space activation by working in the $S_n$ irreducible representation subspace where the ground state energy lies.  This would impact the stability of the numerical simulations, which imply the best-suited models are with 8-16 spins. This bottleneck in computational resource, as shown in Section \ref{s3}, presents no issue for a potential larger-scale implementation on quantum computers. The benchmarked examples with RBM and Group-equivariant Convolutional Neural Network (GCNN) \cite{Roth_2021} are drawn from NetKet \cite{NetKet} tutorial \url{https://www.netket.org/tutorials.html}, which form the baseline comparison. Note that we implemented no explicit global SU(2) or U(1) symmetry for these benchmark algorithms. For numerical simulation of $S_n$-CQA, we perform random initialization of the parameters. We found that the random initialization already returns the energy which is within roughly $10^{-2}$ precision within ED ground state energy and non-oscillating descents around the ED ground state energy comparing with that of GCNN and RBM. This is likely due to the fact that we used the $S_n$-Fourier space activation with real-valued trial wavefunctions with explicit SU(2) symmetry. We record the optimized energy for the $S_n$-CQA ans{\"a}tze every five iterations, and we set the number of alternating layers $p=4$ for the $3 \times 4$ lattice and $p=6$ for the 12-spin Kagome lattice. In the implementation, we shift the Hamiltonian to $\wt{H}(\lambda) = H(\lambda)+ m \operatorname{1}_{d_{\lambda}} $ to ensure $\wt{H}(\lambda )$ is positive semi-definiteness in the $S_n$-irrep specificity by the partition $\lambda = (\lambda_1, \lambda_2)$ with the total spin label $j = (\lambda_1 -\lambda_2)/2$ (Section \ref{s2}), where $m$ is the total number of transpositions. We only take the real (normalized) part of the wavefunction $\operatorname{Re}(\psi) = \psi + \psi^*$. This can be seen as a post-processing step for the realization of $S_n$-CQA on a quantum computer. We use the Nesterov-accelerated Adam \cite{Dozat_2016} for the $S_n$-CQA optimization with hyper-parameters: $\operatorname{betas} = [0.99, 0.999]$. We also utilize NetKet's ED (Exact Diagonalization) result for the comparison with the exact ground state energy. The ground state is additionally calculated in the Young (Schur) basis by diagonalizing the Heisenberg Hamiltonian $H(\lambda)$ in the irrep $\lambda$ where the ground state lives. It is worth mentioning that our optimized ground state is strictly real-valued and has explicitly SU(2) symmetry, offering the missing yet essential physical interpretation. We provide code and Jupiter notebook in open-access on Github in python. The numerical simulations are run in CPU platform with the 9th Gen 1.4 GHz Intel Core i5 processors. 

\begin{figure*}[pt]
\centering
\includegraphics[width=0.45\textwidth]{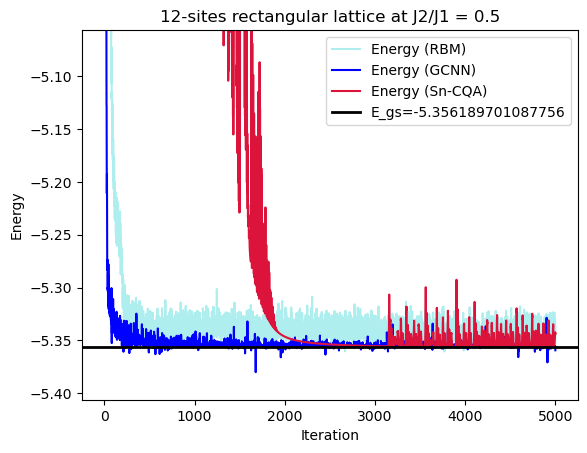}
\includegraphics[width=0.45\textwidth]{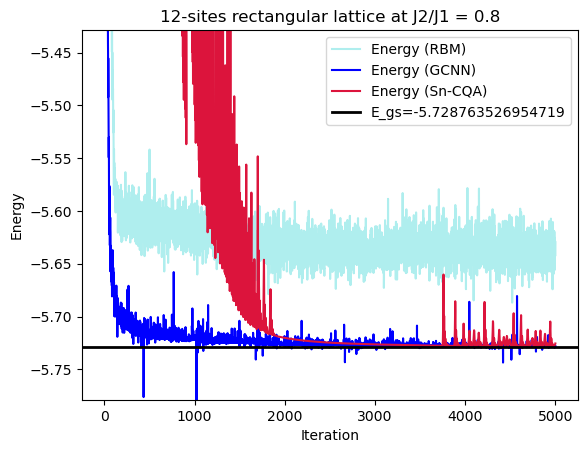}
\includegraphics[width=0.45\textwidth]{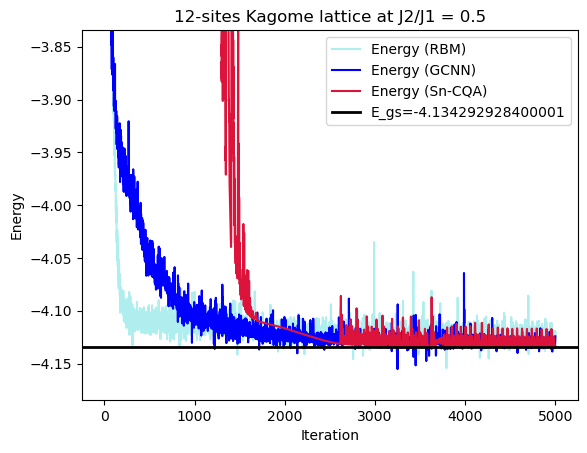}
\includegraphics[width=0.45\textwidth]{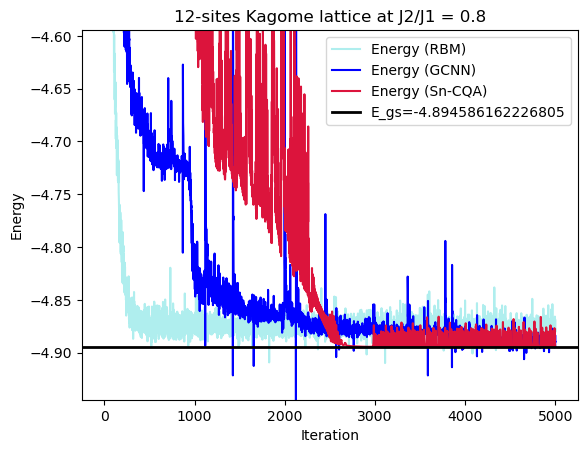}
\caption{For the $3 \times 4$ lattice, in either case, the $S_n$-CQA ans{\"a}tze are able to converge to the ground state at least $10^{-4}$ precision with explicitly reinforced SU($2$) symmetry. This can be seen from the expectation of Sn-CQA never falls below the exact ground state, while non-symmetry respecting algorithms inevitably do. The numerical results are subject to room for further development, for instance-- with better gradient descent algorithms such as to utilize the Hessian--since we have only 200-500 learnable parameters to optimize.   Therefore, we expect the performance and convergence rate of $S_n$-CQA ans{\"a}tze to further increase with perhaps more refined tuning.}
\label{simulation}
\end{figure*}

\subsubsection*{$3 \times 4$ Rectangular Lattice }

In frustrated region of $J_2/J_1 = 0.5, J_2/J_1=0.8$, we found that the ground state lives entirely in the total spin $0$ irrep, corresponding to the partition $\lambda = (4,4)$. In the case of $J_2 = 0.5$, we report that the $S_n$-CQA ans{\"a}tze are able to smoothly converge to the ground state, with error to the exact ground state energy $9.1049e^{-5}$. For $J_2 =0.8$, the $S_n$-CQA returns $5.0587e^{-4}$ precision to the ED ground state energy. We notice that the  $S_n$-CQA seem always to converge to the ground state with reasonable good accuracy without the issue of trapping in local minima, regardless of initialization (random initialization from Gaussian is used). The Learning rate used here is $0.01$. For the GCNN layers in both $J_2$ values we set the feature dimensions of hidden layers $(8,8,8,8)$ and $1024$ samples with the learning rate set for $0.02$. For the RBM model, we fix the learning rate $0.02$ with $1024$ samples.


\subsubsection*{12-Spin Kagome Lattice}

We found by comparing with ED result that the ground state of 12-spin Kagome lattice lives in the total spin $2$ irrep, corresponding to partition $[8, 4]$ in $J_2 =0$, which suggests it to be $5$-fold degenerate. For the both frustration level $J_2/J_1 = 0.5$ and $J_2/J_1 = 0.8$, the ground state lives in total spin 0 irrep, which appear to be non-degenerate. 
We aim to learn the ground state for the $12$-spin Kagome lattice at $J_2/J_1 = 0.5$ and $J_2/J_1 = 0.8$. In the case of $J_2/J_1 = 0.5$, the optimized ground state energy by the $S_n$-CQA ans{\"a}tze at the end of iteration returns $1.5721e^{-4}$ precision to the ED result. In the case $J_2/J_1 = 0.8$, we have the final optimized energy $6.2065e^{-5}$ precision to the ED ground state energy.  The learning rate is set for $0.01$ for the $J_2/J_1 = 0.5$ and $0.8$. We set the GCNN in both frustration points of feature dims $(8,8,8,8)$ with $1024$ samples. The learning rate in both frustrations is set to be $0.02$. The RBM implementation uses $1024$ samples with a learning rate $0.02$ for both cases. 


\section*{Discussion}

In this paper, we introduce a framework to design non-Abelian group-equivariant quantum variational ans{\"a}tze as an example of PQC+ extended from Permutational Quantum Computation (PQC). The restricted universality of the $S_n$-CQA ans{\"a}tze makes it applicable to a wide array of practical problems which would explicitly encode permutation equivariant structure or exhibit global SU($d$) symmetry. Our proof techniques can be used to show the universality of QAOA and verify the four-locality of generic $\operatorname{SU}(d)$ symmetric quantum circuits. Moreover, we illustrate the remarkable efficacy of our approach by finding the ground state of the Heisenberg antiferromagnet $J_1$-$J_2$ spins in a $3 \times 4$ rectangular lattice and 12-spin Kagome lattice in highly frustrated regimes near the speculated phase transition boundaries. We provided strong numerical evidence that our $S_n$-CQA can approximate the ground state with high degree of precision, and strictly respecting SU($2$) symmetry. This opens up new avenues for using representation theory and quantum computing in solving quantum many-body problems.

\subsubsection*{Open Problems}

We conclude with several interesting open problems: $(a)$ We would like to find out the computational power of PQC+. In particular, it is interesting to investigate whether quantum circuits can (in polynomial time) approximate matrix elements of any $S_n$ Fourier coefficients. A natural starting place is perhaps based on the restricted universality of $S_n$-CQA ans{\"a}tze in each $S_n$ irrep by asking if a polynomial bounded number if alternating layers $p$ are able to approximate any matrix element of $S_n$ Fourier coefficients. Or we may further loose the condition by asking if a polynomial bounded number if alternating layers $p$ would form an approximate k-design for subgroups $\operatorname{U}(S^\lambda)$ restricted from $\operatorname{U}(V^{\otimes n})$ when imposing the global $\operatorname{SU}(d)$ symmetry. A detailed study of this question will shed some light on the nature and scope of the prospective quantum advantage. $(b)$ In the SM we show that $S_n$-CQA ans\"atze at large $p$ can simulate certain quantum adiabatic evolution with random path-dependent coupling strengths. It would be important to investigate whether the path-dependent coupling strength parameters $\beta_{kl}$ lead to potential amplitude amplification of the spectral gap in the adiabatic path. In particular, one might need to address the physical dynamics of the random path-dependent coupling strengths. $(c)$ More generally, the quantum speed-up we demonstrated here is inherently connected to the PQC
+. Are there other quantum speed-ups within this framework? In particular, $(b)$ suggests a possible route related to quantum annealing. Another possible route may have to do with measurement-based quantum advantage. For instance, see \cite{qml_pres1}. Therefore, one might want to ask if our $S_n$-CQA ans{\"a}tze have other sources of quantum exponential speed-up.  $(d)$ Another open direction would be to benchmark the performance of the $S_n$-CQA ans{\"a}tze in various Heisenberg models and to implement the $S_n$-CQA ans{\"a}tze on a quantum device. 


\section*{Code Availability}
The codes for the numerical simulation can be found at \url{https://github.com/hanzheng98/Sn-CQA}. The C++ implementation of $S_n$ operations can be found at \url{https://github.com/risi-kondor/Snob2}. Data availability is upon request by emailing  \url{hanz98@uchicago.edu}. 


\section*{Acknowledgement}

H.Z. and Z.L. are contributed equally. We thank Alexander Bogatskiy, Giuseppe Carleo, Jiequn Han, Antonio
Mezzacapo, Horace Pan, Christopher Roth, Hy Truong
Son, Miles Stoudenmire, Kristan Temme, Erik H. Thiede,
Chihchan Tien, Zilong Zhang, Wenda Zhou, and Pei
Zeng for their useful discussions. We especially thank
Liang Jiang and Kanav Setia for suggestions during
the preparation of the manuscript. J.L. is supported in
part by International Business Machines (IBM) Quantum through the Chicago Quantum Exchange, and the
Pritzker School of Molecular Engineering at the University of Chicago through AFOSR MURI (FA9550-21-1-
0209). S.S. acknowledges support from the Royal Society
University Research Fellowship.


\bibliographystyle{apsrev4-1}
\bibliography{cs2.bib}

\pagebreak
\clearpage
\foreach \x in {1,...,\the\pdflastximagepages}
{
	\clearpage
	

\section*{Supplementary material}

# Supplemental Material: Speeding up Learning Quantum States through Group Equivariant Convolutional Quantum Ansätze

Han Zheng,[1, 2, *] Zimu Li,[2, †] Junyu Liu*,[3, 4, 5, ‡] Sergii Strelchuk,[2, §] and Risi Kondor[1, 6, 7, ¶]

[1]*Department of Statistics, The University of Chicago, Chicago, IL 60637, USA*
[2]*DAMTP, Center for Mathematical Sciences, University of Cambridge, Cambridge CB30WA, UK*
[3]*Pritzker School of Molecular Engineering, The University of Chicago, Chicago, IL 60637, USA*
[4]*Chicago Quantum Exchange, Chicago, IL 60637, USA*
[5]*Kadanoff Center for Theoretical Physics, The University of Chicago, Chicago, IL 60637, USA*
[6]*Department of Computer Science, The University of Chicago, Chicago, IL 60637, USA*
[7]*Flatiron Institute, New York City, NY 10010, USA*
(Dated: September 14, 2023)

*: corresponding author.
In this Supplemental Material, we will provide more mathematical backgrounds and prove theorems given in the main text.

## CONTENTS

## I. MORE FACTS ABOUT SCHUR-WEYL DUALITY AND $S_n$ REPRESENTATION

**Lemma I.1.** *Any local part $\hat{\boldsymbol{S}}_i \cdot \hat{\boldsymbol{S}}_j$ constituting a Heisenberg Hamiltonian $H = \sum_{ij} J_{ij} \hat{\boldsymbol{S}}_i \cdot \hat{\boldsymbol{S}}_j$ can be written as*

$$\hat{\boldsymbol{S}}_i \cdot \hat{\boldsymbol{S}}_j = \frac{1}{2}(ij) - \frac{1}{4}I.$$

*Proof.* Let us consider $\hat{\boldsymbol{S}}_1 \cdot \hat{\boldsymbol{S}}_2$. Expanded by definition,

$$\hat{\boldsymbol{S}}_1 \cdot \hat{\boldsymbol{S}}_2 = \frac{1}{2}(J_{12}^2 - J_1^2 - J_2^2),$$

_____

* hanz98@uchicago.edu
† lizm@mail.sustech.edu.cn
‡ junyuliu@uchicago.edu
§ ss870@cam.ac.uk
¶ risi@cs.uchicago.edu

where for now the subscripts on $J^2$ denotes all sites $J^2$ acting on. Under total spin basis, it is easy to see that

$$J_{12}^2 - J_1^2 - J_2^2 = \begin{pmatrix} 2 & 0 & 0 & 0 \\ 0 & 2 & 0 & 0 \\ 0 & 0 & 2 & 0 \\ 0 & 0 & 0 & 0 \end{pmatrix} - \frac{3}{4}I - \frac{3}{4}I = \begin{pmatrix} \frac{1}{2} & 0 & 0 & 0 \\ 0 & \frac{1}{2} & 0 & 0 \\ 0 & 0 & \frac{1}{2} & 0 \\ 0 & 0 & 0 & -\frac{3}{2} \end{pmatrix}.$$

While

$$(12) = \begin{pmatrix} 1 & 0 & 0 & 0 \\ 0 & 1 & 0 & 0 \\ 0 & 0 & 1 & 0 \\ 0 & 0 & 0 & -1 \end{pmatrix}.$$

Therefore,

$$\hat{\boldsymbol{S}}_1 \cdot \hat{\boldsymbol{S}}_2 = \frac{1}{2}(J_{12}^2 - J_1^2 - J_2^2) = \frac{1}{2}((12) - \frac{1}{2}I).$$

This argument holds for any $i, j$, hence the proof follows. $\qquad\qquad\square$

**Theorem I.2.** *Let $\langle \pi_{\mathrm{SU}(d)} \rangle$ be the collection of all matrices on $V^{\otimes n}$ generated by $\pi_{\mathrm{SU}(d)}$ (i.e., taking both linear spans and matrix products) and let $\langle \pi_{S_n} \rangle$ be defined similarly. Note that $\langle \pi_{S_n} \rangle$ is just the representation of $\mathbb{C}[S_n]$. Let $M$ be an arbitrary matrix of $V^{\otimes n}$. If it commutes with $\langle \pi_{\mathrm{SU}(d)} \rangle$, then $M \in \langle \pi_{S_n} \rangle$. If it commutes with $\langle \pi_{S_n} \rangle$, then $M \in \langle \pi_{\mathrm{SU}(d)} \rangle$.*

This theorem is proved by the so-called *double commutant theorem*, details can be found in [1].

**Theorem I.3** (Wedderburn Theorem). *Given any $S_n$-irrep $S^\lambda$, let $\mathrm{End}(S^\lambda)$ denote the collection of linear transformation of $S^\lambda$. With respect to any basis, e.g., the Young basis, $\mathrm{End}(S^\lambda)$ is simply the collection of all $\dim S^\lambda \times \dim S^\lambda$ matrices. As a vector space, the group algebra $\mathbb{C}[S_n]$ is isomorphic with the direct sum of $\mathrm{End}(S^\lambda)$:*

$$\mathbb{C}[S_n] \cong \bigoplus_{\lambda \vdash n} \mathrm{End}(S^\lambda).$$

We emphasize that being different from decomposing $V^{\otimes n}$ by $\mathrm{SU}(d) - S_n$ duality, all kinds of Young diagrams $\lambda$, standing for inequivalent $S_n$-irreps, appears once and only once in the above direct sum decomposing $\mathbb{C}[S_n]$. In any case, restricting to each $\lambda$, $\widetilde{\pi}_{\mathbb{C}[S_n]}$ produces all $\dim S^\lambda \times \dim S^\lambda$ matrices and hence any matrix commuting with $\widetilde{\pi}_{\mathbb{C}[S_n]}$ should be a scalar.

## II. PROOFS OF THE RESTRICTED UNIVERSALITY

Let $\mathfrak{gl}(d, \mathbb{C})$ denote the Lie algebra of the complex general linear group. It is simply the collection of all $d \times d$ complex matrices. Let $\mathfrak{d}(d)$ denote its Cartan subalgebra consisting of all diagonal matrix. We first prove the following lemma in this general setting and later we would set $d = \dim S^\lambda$, or harmlessly $d = \dim S^\lambda \times m_{S^\lambda}$, because our focuses are $S_n$-irreps $S^\lambda$.

**Lemma II.1.** *Let $E_{ij} \in \mathfrak{gl}(d, \mathbb{C})$ denote the matrix unit with entry 1 at $(i, j)$ and 0 elsewhere. Given any $d \times d$ matrix $M$, let $\mathcal{I} \subset \{1, ..., d\} \times \{1, ..., d\}$ be the index set corresponding to nonzero off-diagonal entries $c_{ij}$ of $M$. Then Lie subalgebra generated by $\mathfrak{d}(d)$ and $M$ contains*

$$\mathfrak{d}(d) \oplus \Big( \bigoplus_{(i,j) \in \mathcal{I}} R_{ij} \Big),$$

*where $R_{ij}$ is the 1-dimensional root space spanned by $E_{ij}$.*

*Proof.* For any diagonal matrix $h = \mathrm{diag}(h_1, ..., h_d) \in \mathfrak{d}(d)$, the Lie bracket $[h, E_{kl}] = (h_k - h_l)E_{kl}$ gives the root $a_{kl}(h) := h_k - h_l$ of $E_{kl}$ under $h$. Let $h := \sum_i \lambda_i E_{ii}$ with $\lambda_i$ being determined later. We define

$$M_1 := [h, A] = [\sum_i \lambda_i E_{ii}, \sum_{(k,l) \in \mathcal{I}} c_{kl} E_{kl}] = \sum_i \sum_{(k,l) \in \mathcal{I}} a_{kl}(E_{ii})\lambda_i E_{kl} = \sum_{(k,l) \in \mathcal{I}} (a_{kl} \cdot \lambda) c_{kl} E_{kl},$$

where in the second step of the above computation, we omit all possible diagonal elements of $A$. This is legal because $h$ commutes with any diagonal matrix. Besides, $a_{kl} \cdot \lambda$ is thought as the inner product of $a_{kl} = (a_{kl}(E_{ii}))$ and $\lambda = (\lambda_i)$. Even

most components of $a_{kl}$ equal zero by the definition of roots, it turns out that the notation $a_{kl} \cdot \lambda$ is neat for the following proof. With $M_1$ being defined, we continue to set

$$M_2 := [h, [h, A]] = \sum_{(k,l) \in \mathcal{I}} (a_{kl} \cdot \lambda)^2 c_{kl} E_{kl},$$

$$\vdots$$

$$M_{|\mathcal{I}-1|} := [h, ..., [h, A] \cdots] = \sum_{(k,l) \in \mathcal{I}} (a_{kl} \cdot \lambda)^{|\mathcal{I}-1|} c_{kl} E_{kl}.$$

Recall that $|\mathcal{I}|$ is the number of nonzero off-diagonal elements of $A$. Then let us consider the following *Vandermonde matrix*:

$$V(a_{kl} \cdot \lambda) = \begin{pmatrix} 1 & 1 & \cdots & 1 \\ a_{k_1 l_1} \cdot \lambda & a_{k_2 l_2} \cdot \lambda & \cdots & a_{k_{|\mathcal{I}|} l_{|\mathcal{I}|}} \cdot \lambda \\ \vdots & \vdots & \ddots & \vdots \\ (a_{k_1 l_1} \cdot \lambda)^{|\mathcal{I}|-1} & (a_{k_2 l_2} \cdot \lambda)^{|\mathcal{I}|-1} & \cdots & (a_{k_{|\mathcal{I}|} l_{|\mathcal{I}|}} \cdot \lambda)^{|\mathcal{I}|-1} \end{pmatrix}.$$

Viewing $\mathfrak{gl}(d, \mathbb{C})$ as a $d \times d$-dimensional vector space with $E_{ij}$ as the standard basis, we note that $V(a_{kl} \cdot \lambda)$ transforms *nonzero* vectors

$$\{c_{k_1 l_1} E_{k_1 l_1}, ...., c_{k_{|\mathcal{I}|} l_{|\mathcal{I}|}} E_{k_{|\mathcal{I}|} l_{|\mathcal{I}|}}\} \text{ to } \{M, M_1, ..., M_{|\mathcal{I}-1|}\}.$$

If $\det V(a_{kl} \cdot \lambda) = (-1)^{|\mathcal{I}|(|\mathcal{I}|-1)/2} \prod_{s<t} (a_{k_s l_s} \cdot \lambda - a_{k_t l_t} \cdot \lambda) \neq 0$, then the transformation has an inverse and thus the linear span of $\{M_r\}$ equals that of $\{c_{kl} E_{kl}\}$ which turns out to be $\bigoplus_{(i,j) \in \mathcal{I}} R_{ij}$ by definition.

To show that $\det V(a_{kl} \cdot \lambda) \neq 0$, we simply note that $a_{k_s l_s} - a_{k_t l_t} \neq 0$ as different roots. A basic statement from linear algebra tells us that the union of hyperplanes of finitely-many nonzero vectors, here are $a_{k_s l_s} - a_{k_t l_t}$, cannot cover the whole vector space. Thus we can always find some nonzero $\lambda$ which does not belong to any of these hyperplanes: i.e., $(a_{k_s l_s} \cdot \lambda - a_{k_t l_t} \cdot \lambda) \neq 0$ to fulfill the requirement. On the other hand, the Lie subalgebra by $\mathfrak{d}(d)$ and $M$ contains all possible linear combinations and Lie brackets of these elements and the proof follows. $\qquad \square$

The crucial technique of using Vandermonde matrix is adapted from the classical paper [2] in which Kuranishi proved that any semisimple Lie algebra by merely two elements. Roughly speaking, one of these elements is $h$ in the above proof obtained by solving $\lambda$ and anther is just the summation of all matrix units $E_{ij}$ for $i \neq j$. However, we cannot take the accessibility of these elements for granted on a qudits system. To be precise, we take some irrep $S^\lambda$ (with its equivalent copies) decomposed from the $n$-qudit system. Then we denote by $\mathfrak{d}(S^\lambda)$ the Cartan subalgebra of all diagonal matrices. By Okounkov-Vershik theorem [3, 4], it equals the representation of $GZ_n$ under Young basis *generated* by YJM-elements. To span the Cartan subalgebra and define $h$ however, $n$ YJM-elements are generally not enough and we need to take their high order products $X_i X_j \cdots X_k$. This leads to a practical problem for ansatz design. For another generator mentioned above, we need to design an operator whose matrix representation is the summation of all $E_{ij}$. For qubits, the Hadamard operator $H^{\otimes n}$, even containing $-1$ in its entries, is likely to be chosen at the first glance. However, we are working on Young basis and by Schur-Weyl duality and Wedderburn theorem, $H^{\otimes n}$ is a scalar matrix on any $S_n$-irrep. With these difficulties being clarified, we are going to solve the problem which finally gives us the analytical form of $S_n$-CQA ansatz. We present in the following sections that our methods can also be used to prove the universality of QAOA and verify that any $SU(d)$-invariant quantum circuit can be built from 4-local $SU(d)$-invariant unitaries up to phase factors.

In a qudits system, the matrix $M$ from Lemma II.1 can be chosen as a Hermitian matrix which we interpret as the Hamiltonian of a certain physical system in the main text. To be specific, we consider the following kind of Hermitian matrices.

**Definition II.2.** Given an arbitrary Hermitian matrix $H_P$, let $\mathcal{G}_{H_P}$ be the underlying indices graph of $H_P$. The matrix $H_P$ is *path-connected* if the associated graph $\mathcal{G}_{H_P}$ is connected. For instance, the matrix

$$H_P = (E_{14} + E_{23} + E_{35} + E_{45}) + (E_{41} + E_{32} + E_{53} + E_{54})$$

is path-connected. However,

$$H_P' = (E_{14} + E_{23} + E_{55}) + (E_{41} + E_{32} + E_{55})$$

is disconnected with partition $\{1, 4\}, \{2, 3\}, \{5\}$. Path-connectedness is also seen in the famous Perron-Frobenius theorem which has lots of applications in graph theory.

**Lemma II.3.** *Let $H_P$ be a path-connected Hamiltonian. Then the generated Lie algebra $\langle \mathfrak{d}(d), H_P \rangle = \mathfrak{gl}(d, \mathbb{C})$. Especially, let us consider $\mathfrak{d}_{\mathbb{R}}(d)$ consisting of all real-valued diagonal matrices. Generated over $\mathbb{R}$, $\langle i\mathfrak{d}_{\mathbb{R}}(d), iH_P \rangle_{\mathbb{R}} = \mathfrak{u}(d)$ the Lie algebra of $\mathrm{U}(d)$ consisting of all skew Hermitian matrices.*

*Proof.* By Lemma II.1,

$$\mathfrak{d}(d) \oplus \Big( \bigoplus_{(i,j) \in \mathcal{I}} R_{ij} \Big) \subset \langle \mathfrak{d}(d), H \rangle.$$

Given any pair of indices $i \neq j$, since the underlying indices graph of $H_P$ is path connected, we can always find a path $i = i_0 \to i_1 \to \cdots \to i_k \to i_{k+1} = j$ with *no repeating* indices for which $H_{P,(i_l, i_{l+1})} \neq 0$. Thus

$$E_{ij} = E_{i_0 i_{k+1}} = [\cdots[[E_{i_0 i_1}, E_{i_1 i_2}], E_{i_2 i_3}], \cdots, E_{i_k i_{k+1}}]$$

belongs to the generated Lie algebra. Since $\mathfrak{gl}(d, \mathbb{C}) = \mathfrak{d}(d) \otimes (\bigoplus_{(i,j)} R_{ij})$ for arbitrary $i \neq j$, we prove the first statement.

On the other hand, by definition, the real Lie algebra $\langle i\mathfrak{d}_{\mathbb{R}}(d), iH_P \rangle_{\mathbb{R}}$ is contained in $\mathfrak{u}(d)$. Since its complexification can be expanded as

$$\langle i\mathfrak{d}_{\mathbb{R}}(d), iH_P \rangle_{\mathbb{R}} + i\langle i\mathfrak{d}_{\mathbb{R}}(d), iH_P \rangle_{\mathbb{R}} = \langle \mathfrak{d}(d), H_P \rangle = \mathfrak{gl}(d, \mathbb{C})$$

and since any matrix from $\mathfrak{gl}(d, \mathbb{C})$ is uniquely decomposed as the sum of a Hermitian matrix and a skew-Hermitian matrix, we complete the proof. □

**Theorem II.4.** *Let $H_P$ be any path-connected Hamiltonian and let $X_{\mathrm{GZ}} \in \mathrm{GZ}_n$ denote a real linear combination of elements from $\{X_i, X_i X_j, X_i X_j X_k, ...\}$. Given any $S^\lambda$, the generated group $H := \langle \exp(i\beta X_{\mathrm{GZ}}), \exp(iH_P) \rangle$ equals the unitary group $\mathrm{U}(S^\lambda)$ restricted to $S^\lambda$. Thus the ansätze constructed by alternating exponentials can be used to approximate any eigenstate of $H_P$.*

*Proof.* Before applying Lemma II.1 & II.3, we reminder that $H$ is generated in the pure algebraic sense, it may not even be a Lie subgroup (with smooth structure) and we may not be able to talk its Lie algebra. One strategy to solve this problem is considering the closure $\bar{H}$ and employing Cartan's closed subgroup theorem. However, there is a another powerful theorem on Lie groups suitable for the current salutation: Yamabe's theorem [5–7] which states that any arcwise-connected/path-connected subgroup of a Lie group is itself a Lie subgroup. In our setting, $H$ is path-connected as any two points in the $H$ are path-connected by products of exponentials. Thus we can talk about its Lie algebra $\mathfrak{h}$. By definition, $\exp(i\beta X_{\mathrm{GZ}}), \exp(i\gamma H_P) \in H$ and thus

$$iX_{\mathrm{GZ}} = \frac{d}{d\theta}\Big|_{\theta=0} e^{i\theta X_{\mathrm{GZ}}}, \quad iH_P = \frac{d}{d\gamma}\Big|_{\gamma=0} e^{i\gamma H_P} \in \mathfrak{h}.$$

Thus Lemma II.3 tells us that $\mathfrak{h}$ equals $\mathfrak{u}(S^\lambda)$. Since $\mathrm{U}(S^\lambda)$ is a connected Lie group and $H$ is generated by $\mathfrak{u}(S^\lambda)$, They must be identical with $\mathrm{U}(S^\lambda)$. □

We would discuss more about applying Yamabe's theorem in Section IV. Before presenting the following materials, especially for Lemma II.6, we should explain more about how to convert the language of root systems used Lemma II.1 into that of content vectors from $S_n$-representation theory. For any $X \in \mathrm{GZ}_n$, its root $a_{kl}(X)$ defined under matrix unit $E_{kl}$ is simply its $k$th diagonal element minus $l$th diagonal element. For instance, for any YJM-element $X_i$, $a_{kl}(X_i) = \alpha_{T_k}(i) - \alpha_{T_l}(i)$ where $T_k, T_l$ denote the $k$the and $l$th standard Young tableaux/Young basis elements respectively for irrep space $S^\lambda$. On the other hand, we see in proving Theorem II.4, $X_{\mathrm{GZ}}$ is taken as an arbitrary linear combination of elements from $\{X_i, X_i X_j, X_i X_j X_k, ...\}$ which spans the corresponding $\mathfrak{d}(S^\lambda)$. In design quantum circuits, we hope to bound this spanning set. A reckless restriction to simple YJM-elements $\{X_i\}$ violates the condition

$$a_{ij} - a_{kl} = (\alpha_{T_i} - \alpha_{T_j}) - (\alpha_{T_k} - \alpha_{T_l}) \neq 0$$

instantiated in the following diagrams

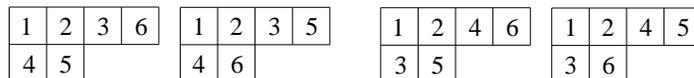

FIG. 1. Four standard tableaux/Young basis elements/content vectors of $(4,2)$-irrep of $S_6$ with $d = 2$. Differences between the first two and the last two content vectors are the same.

FIG. 2. Four standard tableaux/Young basis elements/content vectors of $(3,2,1)$-irrep of $S_9$ with $d=3$. Differences between the first two and the last two content vectors are the same.

Fortunately, restricting to second order YJM-elements $\{X_i, X_i X_j\}$ suffices to establish Lemma II.1.

**Definition II.5.** For any standard Young tableau $T$ of shape $\lambda = (\lambda_1, ..., \lambda_d)$, Measuring via second-order products $X_k X_l$ of YJM-elements, we obtain an $(n \times n)$-component vector denoted by $\beta_T^{(2)}(r, s)$. It can also be written as a real symmetric matrix $(\beta_T)^T \beta_T$ and hence we call it *tensor product content vector* or *second order content vector*.

**Lemma II.6.** *Different standard Young tableaux/Young basis elements have different tensor product content vectors. Moreover, let $T_i, T_j, T_k, T_l$ be four standard Young tableaux of size $n$. We require that $T_i \neq T_j$ and $T_k \neq T_l$. If all underlying Young diagrams are the same, we further require that $T_i \neq T_k$ or $T_j \neq T_l$ (this requirement comes from the computation of roots of $E_{ij}, E_{kl}$ above). Then*

$$\beta_i^{(2)} - \beta_j^{(2)} \neq \beta_k^{(2)} - \beta_l^{(2)},$$

*where $\beta_i^{(2)}$ is an abbreviation for $\beta_{T_i}^{(2)}$. That is, when extended to tensor product content vectors, the differences of different standard Young tableaux are still different.*

*Proof.* The first statement follows by the one-to-one correspondence between content vectors and spectra of Young basis [3, 4, 8]. Actually, we abuse the language above for simplicity: spectra of Young basis, as the name indicates, means eigenvalues under YJM-elements and are just what we denote by $\alpha_T$. While content vectors simply refer to collections of ordered numbers read off from standard Young tableau. Okounkov and Vershik proved these two concepts are equivalent in a more abstract setting. To prove the second statement, we first look at differences between common content vectors

$$\beta_i - \beta_j, \quad \beta_k - \beta_l.$$

If they are different, then we finish the proof.

Suppose then $\beta_i - \beta_j = \beta_k - \beta_l$. Since $\beta_k - \beta_l \neq 0$,

$$\beta_i - \beta_j = \beta_k - \beta_l \neq \beta_l - \beta_k.$$

Thus there must be some integer $r \in \{1, ..., n\}$ such that

$$\beta_i(r) - \beta_j(r) \neq \beta_l(r) - \beta_k(r) \Leftrightarrow \beta_i(r) + \beta_k(r) \neq \beta_j(r) + \beta_l(r).$$

If $\beta_i(r) - \beta_k(r) = \beta_j(r) - \beta_l(r) \neq 0$, then

$$\beta_i^{(2)}(r,r) - \beta_k^{(2)}(r,r) = \Big(\beta_i(r) - \beta_k(r)\Big)\Big(\beta_i(r) + \beta_k(r)\Big)$$
$$\neq \Big(\beta_j(r) - \beta_l(r)\Big)\Big(\beta_j(r) + \beta_l(r)\Big) = \beta_j^{(2)}(r,r) - \beta_l^{(2)}(r,r)$$
$$\implies \beta_i^{(2)}(r,r) - \beta_j^{(2)}(r,r) \neq \beta_k^{(2)}(r,r) - \beta_l^{(2)}(r,r) \implies \beta_i^{(2)} - \beta_j^{(2)} \neq \beta_k^{(2)} - \beta_l^{(2)}.$$

To tackle general cases, let us simply denote by $r$ the integers from $\{1, ..., n\}$ with $\beta_i(r) + \beta_k(r) \neq \beta_j(r) + \beta_l(r)$ and let $r'$ denote these with $\beta_i(r') + \beta_k(r') = \beta_j(r') + \beta_l(r')$. Since $\beta_i - \beta_j = \beta_k - \beta_l$,

$$\beta_i(r') + \beta_k(r') = \beta_j(r') + \beta_l(r') \Leftrightarrow \beta_i(r') - \beta_j(r') = \beta_l(r') - \beta_k(r') = \beta_k(r') - \beta_l(r') = 0.$$

As $\beta_i - \beta_j = \beta_k - \beta_l$ is a nonzero vector by the requirements $i \neq j$ and $k \neq l$ and by Okounkov-Vershik theorem,

$$\beta_i(r) - \beta_j(r) = \beta_k(r) - \beta_l(r) \neq 0, \text{i.e.}, \beta_i(r) \neq \beta_j(r), \beta_k(r) \neq \beta_l(r).$$

Now suppose $\beta_i(r) - \beta_k(r) = \beta_j(r) - \beta_l(r) = 0$ for all $r$. We can no longer say that $\beta_i^{(2)}(r,r) - \beta_j^{(2)}(r,r) \neq \beta_k^{(2)}(r,r) - \beta_l^{(2)}(r,r)$. However, since $T_i \neq T_k$ (otherwise, $T_j$ must be different from $T_l$ and that finishes the proof), there must be some $s$ such that $\beta_i(s) \neq \beta_k(s)$. Again with the assumption $\beta_i - \beta_j = \beta_k - \beta_l$, $\beta_j(s) \neq \beta_l(s)$. Since $\beta_i(r) \neq \beta_j(r)$,

$$\beta_i(r)\beta_i(s) - \beta_k(r)\beta_k(s) = \beta_i(r)\Big(\beta_i(s) - \beta_k(s)\Big)$$
$$\neq \beta_j(r)\Big(\beta_i(s) - \beta_k(s)\Big) = \beta_j(r)\Big(\beta_j(s) - \beta_l(s)\Big) = \beta_j(r)\beta_j(s) - \beta_l(r)\beta_l(s)$$
$$\implies \beta_i^{(2)}(r,s) - \beta_j^{(2)}(r,s) = \beta_i(r)\beta_i(s) - \beta_j(r)\beta_j(s)$$
$$\neq \beta_k(r)\beta_k(s) - \beta_l(r)\beta_l(s) = \beta_k^{(2)}(r,s) - \beta_l^{(2)}(r,s).$$

This finishes the proof. $\qquad\square$

*Remark.* We now examine the complex Lie algebra generated by $Y_k Y_l$ and $H_P$. As a caveat, $\{Y_k Y_l\}$ is not enough to generate the whole $GZ_n$ as the Cartan subalgebra, but with Lemma II.6, we are allowed to perform the same proof method from Lemma II.1. For instance, we can rewrite $M_1$ defined there as follows:

$$[\sum_{i,j} \lambda_{ij} Y_i Y_j, H_P] = [\sum_{i,j} \lambda_{ij} Y_i Y_j, \sum_{(k,l)\in\mathcal{I}} c_{kl} E_{kl}] = \sum_{i,j} \sum_{(k,l)\in\mathcal{I}} (\beta_k^{(2)}(i,j) - \beta_l^{(2)}(i,j))\lambda_{ij} c_{kl} E_{kl}$$
$$= \sum_{(k,l)\in\mathcal{I}} \big((\beta_k^{(2)} - \beta_l^{(2)}) \cdot \lambda\big) c_{kl} E_{kl}.$$

Accordingly, we can find a nonzero solution to $\lambda$ such that iterated Lie brackets of $h := \lambda_{ij} Y_i Y_j$ and $H_P$ generate 1-dimensional subspaces corresponding to nonzero entries of $H_P$.

Then by Lemma II.3, $Y_k Y_l, H_P$ are able to generate all the off-diagonal root spaces $R_{ij}$ if $H_P$ is path-connected. To complement the missing diagonal elements which cannot be generated by second order YJM-elements, we note that $[E_{ij}, E_{ji}] = E_{ii} - E_{jj}$. Thus all traceless diagonal matrices can be found in $\langle Y_k Y_l, H_P \rangle$. Moreover, since for instance, $Y_2^2 = Y_2 Y_2$ is diagonal with only positive entries, all diagonal matrices can be generated. Thus $\langle Y_k Y_l, H_P \rangle = \mathfrak{gl}(S^\lambda)$ and by the second part of Lemma II.3, $\langle i Y_k Y_l, i H_P \rangle_{\mathbb{R}} = \mathfrak{u}(S^\lambda)$. In conclusion, we obtain the following theorem.

**Theorem II.7.** *The subgroup generated by $Y_k Y_l$ with a path-connected Hamiltonian $H_P$ still equals $\mathrm{U}(S^\lambda)$ when restricted to any $S^\lambda$. Since $Y_k Y_l = (X_k + kI)(X_l + kI) = X_k X_l + k(X_k + X_l) + k^2 I$ and since $\exp(i\theta I)$ is simply a phase term, a $S_n$-CQA ansatz is written as*

$$\cdots \exp(-i\sum_{k,l} \beta_{kl} X_k X_l) \exp(-i\gamma H_P) \exp(-i\sum_{k,l} \beta'_{kl} X_k X_l) \exp(-i\gamma' H_P) \cdots,$$

*where we redefine $X_1$ as $I$ with which any first-order YJM-element $X_i$ can be written as $X_i X_1$.*

*Remark.* In the above remark, we argue that $\langle Y_k Y_l, H_P \rangle = \mathfrak{gl}(S^\lambda)$ and hence its real form equals $\mathfrak{u}(S^\lambda)$. This is true for any *fixed* $S^\lambda$ with its equivalent copies decomposed from $V^{\otimes n}$. As the main topic in Section III, we are going to generate $\bigoplus_\lambda \mathfrak{gl}(S^\lambda)$ which encompasses all inequivalent irreps simultaneously.

It is straightforward to see that all the above proofs hold in this general case except one point: one may not be able to manipulate phase factors of unitaries generated in inequivalent $S^\lambda$ arbitrarily. We would like to call this a *phase factor problem*. To be precise, exponentials of second order YJM elements give us access to change phases, e.g., $e^{-itI}$. But as we are concerning more than one inequivalent irrep blocks, phases of unitaries from inequivalent irrep blocks may not vary independently. Even though, since the Lie brackets $[E_{ij}, E_{ji}$ in Lemma II.3 yield all kinds of traceless diagonal matrices, we are still safe to claim:

$$\bigoplus_\lambda \mathfrak{sl}(S^\lambda) \subset \langle Y_k Y_l, H_P \rangle, \quad \bigoplus_\lambda \mathfrak{su}(S^\lambda) \subset \langle i Y_k Y_l, i H_P \rangle_{\mathbb{R}}.$$

Restricting to second order YJM-elements causes no troubles if we ignore these phase factors. More discussions can be found in Section III & Section IV.

## III. FOUR-LOCALITY OF $\mathrm{SU}(d)$-SYMMETRIC QUANTUM CIRCUITS

It has been shown in [9–11] that to build an $\mathrm{SU}(d)$-invariant circuit for $d \geq 3$, merely taking 2-local $\mathrm{SU}(d)$-invariant unitaries may not be sufficient. This contrasts with the well-known fact that 2-unitaries are universal for quantum circuits without any symmetry [12]. We now explain our solution by employing 4-local Hamiltonians. We follow the notation in [11, 13, 14] which denote by $\mathcal{V}_k$ the subgroup generated by $k$-local $\mathrm{SU}(d)$-invariant unitaries.

*Remark.* We first discuss a subtle point about the definition of $\mathcal{V}_k$ which is important for the following proof. It is also clarified in [11, 13]. By definition, elements from $\mathcal{V}_k$ are of the form

$$\exp(-it_1 H_1) \cdots \exp(-it_p H_p), \tag{1}$$

where by Schur-Weyl duality, the Hamiltonians are representations of elements from the group algebra $\mathbb{C}[S_n]$. Moreover for each $H_q$, there should be a *fixed* subcollection of $k$ qudits where $H_q$ is supported on. This guarantees that each $\exp(-it_q H_q)$ is still $k$-local.

On the other hand, one can loose the condition by taking Hamiltonians $H'_q = \sum_i c_i \widetilde{\pi}(\sigma_i)$ such that different permutations $\sigma_i$ can be supported on different $k$ sites. Then Let us consider the *integral group* $\mathcal{V}'_k$ consisting of elements of the form like Eq.(1) only by replacing $H_q$ with $H'_q$. Obviously by definition, $\mathcal{V}_k \subset \mathcal{V}'_k$. Moreover, by taking derivative Eq.(1), one can check by the same methods from Theorem II.4 that $\mathcal{V}_k$ and $\mathcal{V}'_k$ admit the same Lie algebra. As connected Lie groups, $\mathcal{V}_k = \mathcal{V}'_k$ [7] and two of them are all adopted as the definition of quantum circuits generated by $k$ local unitaries.

To address the *phase factor problem* mentioned in the previous section as well as in the main text, we also consider the subgroup $\mathcal{SV}_k \subset \mathcal{V}_k$ with trivial phase factor relative to each $S^\lambda$. To be precise, recall that unitaries from $\mathcal{V}_k$ respecting $\mathrm{SU}(d)$ symmetry can be decomposed as a sum of unitary blocks by Schur-Weyl duality. We require that each unitary block has determinant equal to 1 and we say these $\mathrm{SU}(d)$-symmetric quantum circuits are taken with phase factors being ignored. We now prove Lemma III.1 which gives us a 2-local path-connected Hamiltonian. Combing with the fact that second order YJM-elements $X_k X_l$ are 4-local, we obtain Theorem III.2.

**Lemma III.1.** *For any qudits system, the Hamiltonian $H_S = \sum_{i=1}^{n-1}(i, i+1)$ defined by summation of adjacent transpositions is path-connected under both computational basis and Young basis.*

*Proof.* Before we present the proof, we reminder that change of bases may break the connectedness of the indices graph. For instance, the $2 \times 2$ Pauli $X$ matrix is path-connected. However, $YXY^{-1} = YXY^\dagger = Z$ which is not path-connected. When $d = 2$, the unfrustrated antiferromagnetic Heisenberg Hamiltonian $H = \sum_{i=1}^n \hat{S}_i \cdot \hat{S}_{i+1}$ *with periodic boundary condition* is verified to be path-connected under computational basis and the fact is further used in proving Marshall-Lieb-Mattis theorem [15, 16]. We are going to show a general case for any $d$, but under Young basis, so the change of bases is just the Schur transform. We will again use $S_n$ representation theory and there is no need to write out Schur transform explicitly, which would make the proof more intricate.

As a caveat, we consider the summation of adjacent transpositions *without boundary condition*. Specifically, let $S^\lambda$ denote a certain $S_n$ irrep, we want to show that the indices graph $\mathcal{G}$ of $H = \sum_{i=1}^{n-1}(i, i+1)$ is connected. In graph theory, this is equivalent to say that $H$ is *irreducible*. That is, $H$ cannot be block diagonalized under any permutation action. Since we are working under Young basis, matrix entries of any adjacent transposition $(i, i+1)$ can be explicitly written out by the so-called Young's orthogonal form [4, 17]. Roughly speaking, each matrix representation consists of blocks like

$$\begin{pmatrix} \frac{1}{r} & \sqrt{1 - \frac{1}{r^2}} \\ \sqrt{1 - \frac{1}{r^2}} & -\frac{1}{r} \end{pmatrix},$$

where the number $r$ is determined by content vectors with respect to different Young basis elements. Different $(i, i+1)$ would contain different numbers of such blocks located in different positions. Even though, a notable feature is that any off diagonal element of $(i, i+1)$ are nonnegative and so is the summation $H$. Therefore, if $H_S$ is block diagonal in $S^\lambda$, all these adjacent transpositions will be block diagonal simultaneously at *the same* positions. As adjacent transpositions generate $S_n$, the irrep $S^\lambda$ would be decomposed further into these blocks which leads a contradiction.

As a reminder, we do not include the boundary condition $(1, n)$ as it cannot be expressed by Young's orthogonal form. After expanded by products adjacent transpositions, it may bear negative off-diagonal entries and invalidate the above proof method. □

**Theorem III.2.** *Ignoring phase factors,* $\mathrm{SU}(d)$-*invariant quantum circuits can be generated by 4-local* $\mathrm{SU}(d)$-*invariant unitaries. In group-theoretical notation,* $\mathcal{SV}_4 = \mathcal{SV}_n \subset \mathrm{CQA}$.

*Proof.* By definition, $X_i X_j$ consists of at most 4-local unitaries and thus by Lemma III.1 and Theorem II.4, CQA defined by $H_S$ is contained in $\mathcal{V}'_4$ introduced in the previous remark. Note that no matter we restrict to one $S_n$-irrep or consider all of them decomposed from $V^{\otimes n}$ simultaneously, Lemma II.1 & Lemma II.3 always hold. Combining with Lemma II.6 and Theorem II.7, $\mathcal{SV}_n \subset \mathrm{CQA}$. Note that as CQA contains generators with nontrivial phase: e.g., $e^{i\theta I}$, the inclusion is proper. Then by the previous remark

$$\mathcal{SV}_n \subset \mathrm{CQA} \subset \mathcal{V}'_4 = \mathcal{V}_4$$

which implies $\mathcal{SV}_n = \mathcal{SV}_4$ by definition. □

*Remark.* We make a brief remark on the problem of using 2-local $SU(d)$-symmetric unitaries. By Schur-Weyl duality, these unitaries are exactly exponentials of SWAPs, i.e., transpositions $(i, j)$ taken from $S_n$ and [9, 10] provide a through treatment on the Lie algebra $\mathfrak{g}_\lambda$ (hence Lie group) generated by transpositions. For Young diagrams $\lambda$ with more than two rows, there are exceptions, like these in Fig.2, for which $\mathfrak{g}_\lambda \cong \mathfrak{so}(S^\lambda)$ or $\mathfrak{sp}(S^\lambda) \subsetneq \mathfrak{sl}(S^\lambda)$. Thus being different from Lemma II.3, its compact real form can never be $\mathfrak{su}(S^\lambda)$ and hence the generated group is properly contained in $SU(S^\lambda)$. As we explained in the main text, Young diagrams with more than two rows appears in general qudits system with $d \geq 3$. Therefore, even ignoring phase factors, 2-local $SU(d)$ unitaries cannot generate all $SU(d)$ symmetric quantum circuits now. As Lemma II.6 works for any Young diagram, those exceptions do not influence our claims on CQA.

At the end of this section, we examine the inclusion relationship of CQA, $\mathcal{V}_4$ and $\mathcal{V}_n$. Due to the phase factor problem, they are not identical in general. To illustrate this point explicitly, we apply the following proposition from $S_n$-representation theory [4, 17]. We first introduce a related definition:

**Definition III.3.** Let $\lambda = (\lambda_1, ..., \lambda_r)$ be an arbitrarily Young diagram of $n$ boxes. A permutation $\sigma \in S_n$ is of *cycle type* $\lambda$ if it is decomposed into cycles of lengths $\lambda_1, ..., \lambda_r$.

For instance, $\sigma = (134)(56) \in S_6$ is of the type $\lambda = (3, 2, 1)$. The trivial permutation $\mathrm{id} \in S_n$, i.e., the identity map, is of type $(1, ..., 1)$. Products of transpositions like $(i, j)(k, l) \cdots (s, t)$ are of type $(2, 2, ..., 2, 1, ..., 1)$. Higher order cycles like $(i_1, i_2, ..., i_k)$ are of type $(k, 1, ..., 1)$. Their summations are used in [14] to establish constraints with which the concerned Hamiltonians can be generated by transpositions.

**Proposition III.4.** Let $c_\lambda \in \mathbb{C}[S]$ be defined as the sum of all $\sigma \in S_n$ with cycle type $\lambda$. Consider all possible $\lambda$ of $n$ boxes, the collection $\{c_\lambda\}$ forms a basis for the *center* $Z(\mathbb{C}[S_n])$ of $\mathbb{C}[S_n]$.

Since $c_\lambda$ commutes with any $\sigma \in S_n$, its matrix representation, still denoted by $c_\lambda$ for simplicity, under any $S_n$ irrep is a scalar. Except this, the identity $\Pi_\lambda = \frac{1}{n!} \sum_{g \in S_n} \chi_\lambda(g) \rho_\lambda(g)$ using group characters $\chi_\lambda$ [1, 17] forms another (orthonormal) basis because $\Pi_\lambda$ are natural projections into inequivalent irreps blocks. Therefore, by the previous proposition, we can manipulate phases factors of unitaries in these inequivalent $S^\lambda$ by $\{e^{-itc_\lambda}\}$ or $\{e^{-it\chi_\lambda}\}$. However by definition, four locality only gives us access to these $c_\lambda$ with

$$\lambda = (2, 1, ..., 1), (2.2, 1, ..., 1), (3, 1, ..., 1), (4, 1, ..., 1).$$

In the case of CQA, it is easy to check that second order YJM-elements cannot even generate $c_\lambda$ with $\lambda = (4, 1, ..., 1)$. Hence we argue that CQA $\neq \mathcal{V}_4 \neq \mathcal{V}_n$ generally.

## IV. COMPACTNESS OF CQA

As a starting point to unravel more deep facts about the structure of CQA, we prove its compactness. Proofs in this section are independent from others and can be skipped with no harm.

We first remark the difference in proving that $\mathcal{V}_k$, CQA are Lie subgroups. In [13], the author argued that for any fixed $k$ sites, all $k$-local $SU(d)$-invariant unitaries on these sites generate a compact Lie subgroup. This can be seem by applying Schur-Weyl duality to the concerned $k$-qudit subsystem. Since $\mathcal{V}_k$ is generated by $\binom{n}{k}$ many compact subgroups of this kind, it is still a compact Lie subgroup (see more details in [13]). In our case, generators of CQA only form a proper subset of the collection of 4-local unitaries and hence we cannot use Schur-Weyl duality directly but turn to Yamabe's theorem for help (see Theorem II.4). Even though, it is still unclear whether CQA is compact or not, but we are going to provide a affirmative answer.

To this end, let $H_P$ be the problem Hamiltonian used to construct CQA. Let $\mathrm{Tr}_\lambda(H_P)$ denote the trace of $H_P$ restricted in $S^\lambda$. Then the operator $c_{H_P}$ consisting of block matrices $(\mathrm{Tr}_\lambda(H_P) / \dim S^\lambda) I_\lambda$ is a center element from $Z(\mathbb{C}[S_n])$. Similarly, we define $c_{X_k X_l}$ for second order YJM-elements. Recall that $\langle X_k X_l, H_P \rangle$ are able to generate all traceless block matrices with respect to the decomposition of the $n$-qudit system under $SU(d)$ symmetry. Matrices like $H_P - c_{H_P}$, $X_k X_l - c_{X_k X_l}$, and thus $c_{H_P}, c_{X_k X_l}$ can be all found in $\langle X_k X_l, H_P \rangle$. Then we strengthen the discussion from the remark after Theorem II.7 as:

**Lemma IV.1.** *The complex Lie algebra $\langle X_k X_l, H_P \rangle$ equals $\bigoplus_\lambda \mathfrak{sl}(S^\lambda)$ plus the subspace spanned by $c_{H_P}, c_{X_k X_l}$. Accordingly, its real form $\langle iX_k X_l, iH_P \rangle_\mathbb{R}$ equals $\bigoplus_\lambda \mathfrak{su}(S^\lambda)$ plus the real subspace spanned by $ic_{H_P}, ic_{X_k X_l}$.*

*Proof.* By Lemma II.3, Lemma II.6 and the above definition of $c_{H_P}, c_{X_k X_l}$, these elements can be generated by $X_k X_l, H_P$. To prove the converse inclusion, let us take $A \in \langle X_k X_l, H_P \rangle$ arbitrarily. By definition, we write $A = B + C$ where $B$ is a linear combination of $X_k X_l, H_P$ and $C$ stands for a summation of involved Lie brackets. Note that with respect to any $S^\lambda$ decomposed from $V^{\otimes n}$, Lie brackets yield only traceless matrices. Besides, we decompose the matrix $B$ by

$$H_P = (H_P - c_{H_P}) + c_{H_P}, \quad X_k X_l = (X_k X_l - c_{X_k X_l}) + c_{X_k X_l}$$

Its traceless part plus the matrix $C$ can be found in $\bigoplus_\lambda \mathfrak{sl}(S^\lambda)$ while the part with nonzero trace is automatically spanned by $c_{H_P}, c_{X_k X_l}$ by definition. The case for $\langle iX_k X_l, iH_P \rangle_\mathbb{R}$ follows immediately. $\qquad \square$

By definition, the Lie group corresponding to $\bigoplus_\lambda \mathfrak{su}(S^\lambda)$ is just $\mathcal{SV}_n = \mathcal{SV}_4$. On the other hand, let $H$ denote the integral Lie group (see Section III) generated by exponentials of $ic_{H_P}, ic_{X_k X_l}$. It is compact as when restricted to any $S^\lambda$, $H \cong S^1 \cong \mathrm{U}(1)$. As a caveat, integral groups are generally non-compact, e.g., for an irrational number $a$,

$$\begin{pmatrix} e^{-it} & \\ & e^{-iat} \end{pmatrix} \in \mathrm{U}(2)$$

traces a non-compact *irrational curve*. Then the compactness of CQA can be shown by either the theory of algebraic group [18, 19], or by the simple observation: for any $(g, h) \in \mathcal{SV}_4 \times H$, let

$$f : \mathcal{SV}_4 \times H \to \mathrm{CQA}$$

be defined by $f(g, h) = gh$. Note that $h$ commutes with $g$ by definition, thus the map $f$ is a Lie group homomorphism with $\mathrm{Im}\, f = \mathrm{CQA}$ because they share the same Lie algebra by definition. Since $\mathcal{SV}_4 \times H$ is compact, the image is also compact.

**Theorem IV.2.** *The group CQA is connected and compact.*

## V. UNIVERSALITY OF QAOA

The first proof of the universality of the QAOA ansätze was given in [20], where the authors considered problem Hamiltonian of the first-order and second-order nearest-neighbor interaction. [21] subsequently generalized the result to broader families of ansätze defined by sets of graphs and hyper-graphs. We now describe a new proof based on the techniques developed in Section II. We reuse the method in Section II on Pauli $Z$ matrices under computational basis $\{|e_i\rangle\}$. First, we need to redefine Cartan subalgebra and roots. It is well-known that all Pauli operators like $\tilde{Z}_{r_1 \ldots r_s} := Z_{r_1} \otimes \cdots \otimes Z_{r_j}$ are able to span any diagonal matrix under $\{|e_i\rangle\}$ which in turn forms the Cartan subalgebra of $\mathfrak{gl}(\mathbb{C}, 2^n)$. We denote by $\alpha_i(j)$ the eigenvalue of $|e_i\rangle$ measured by $Z_j$. Let $E_{kl}$ be the matrix unit as in Lemma II.1 but under $\{|e_i\rangle\}$. The root $a_{kl}$ of $E_{kl}$ is given by

$$a_{kl}(j) E_{kl} = \mathrm{ad}_{Z_j}(E_{kl}) = (\alpha_k(j) - \alpha_l(j)) E_{kl}.$$

High order components of $a_{kl}$ under $\tilde{Z}_{r_1 \ldots r_s}$ are obtained by products. With all these substitutions, Lemma II.1, Lemma II.3 and Theorem II.4 follow straightforward. We now examine that a local Hamiltonian $H_Z = \sum \beta_{kl} Z_k Z_l$ consisting of first- and second-order product of Pauli $Z$ ($Z_0 = I$) also guarantees the universality. With the same sprint of Lemma II.6, we have

**Lemma V.1.** *Let $v$ be any computational basis element. Then for any $u_1 \neq v_1, u_2 \neq v_2$ (if $u_1 = u_2$, we require $v_1 \neq v_2$ and vice versa),*

$$\alpha_{u_1}^{(2)} - \alpha_{v_1}^{(2)} \neq \alpha_{u_2}^{(2)} - \alpha_{v_2}^{(2)}.$$

*Proof.* One can take this lemma as a corollary of Lemma II.6 where content vectors are replaced by eigenvalues of $Z_j$. Since they only take values from $\{\pm 1\}$, there is a more straightforward way towards the proof by induction. Suppose the statement is true for $n - 1$ qubits. Let $u_1, v_1, u_2, v_2$ be distinct computational basis elements from the $n$-qubit system. We denote by $0, 1$ the spin basis of each qubit. Then possible assignments of their spins (not eigenvalues) of the last qubit are listed as:

$$\begin{array}{cccc}
(0, 0, \quad 0, 0), & (0, 0, \quad 1, 1), & (1, 1, \quad 0, 0), & (1, 1, \quad 1, 1), \\
(0, 1, \quad 0, 1), & (1, 0, \quad 1, 0), & (0, 1, \quad 1, 0), & (1, 0, \quad 0, 1), \\
(0, 0, \quad 0, 1), & (0, 0, \quad 1, 0), & (0, 1, \quad 0, 0), & (1, 0, \quad 0, 0), \\
(0, 1, \quad 1, 1), & (1, 0, \quad 1, 1), & (1, 1, \quad 0, 1), & (1, 1, \quad 1, 0).
\end{array}$$

The last ten cases can be confirmed immediately as differences of eigenvalues measured by $Z_n$ on the last qubit are different. We need to analyze the cases like $(0, 0, \ 1, 1)$. Deleting the last qubits on both sides, suppose then $u_1 = u_2$ and $v_1 = v_2$ as strings of 0 and 1. Since we require $u_1 \neq v_1, u_2 \neq v_2$, there exists some $i < n$ such that $\alpha_{u_1}(i) - \alpha_{v_1}(i) = \alpha_{u_2}(i) - \alpha_{v_2}(i) \neq 0$. Then their eigenvalues of $Z_i Z_n$ are different. If $u_1 \neq u_2$ or $v_1 \neq v_2$ after deleting the last qubits, then by the above induction hypothesis, we obtain the proof.. $\qquad\square$

Finally we have:

**Theorem V.2.** *Let $H$ be any path-connected Hamiltonian under computational basis, the QAOA-ansatz generated by $H, H_Z = \sum \beta_{kl} Z_k Z_l$ is dense in $\mathrm{U}(2^n)$, i.e., it is universal.*

For concrete examples of path-connected Hamiltonian under computational basis, we mentioned in Lemma III.1 that the Heisenberg Hamiltonian $H = \sum_{i=1}^n \hat{S}_i \cdot \hat{S}_{i+1}$ with periodic boundary condition or the uniform summation of Pauli $X$ operators $H_X$ are all possible candidates.

# VI. CORRESPONDENCE BETWEEN SPIN LABELS AND CONTENT VECTORS

**Lemma VI.1.** *Under* $\mathrm{SU}(d) - S_n$ *duality, sequential coupling Casimir operators commute with YJM-elements.*

*Proof.* We prove the case for $\mathrm{SU}(2)$ with only one quadratic Casimir operator $J^2$ for simplicity. General case for $\mathrm{SU}(d)$ follows immediately by considering all the $d-1$ Casimir operators constructed in [22] simultaneously.

The $k$th sequential coupling of $J^2$ on the subsystem of first $k$ qubits is $J_k^2 = (S_x^k)^2 + (S_y^k)^2 + (S_z^k)^2$, where $S_x^k, S_y^k, S_z^k$ are spin operators for $x, y, z$ direction. Thus $J_k^2$ also commutes with $\pi_{\mathbb{C}[S_k]}$ in the subsystem. On the other hand, as the Casimir operator, $J_k^2$ also commutes with $U^{\otimes k} \otimes I^{\otimes(n-k)}$ for any $U \in \mathrm{SU}(2)$. Thus its matrix representation can be expanded by $\mathbb{C}[S_k]$ by Theorem I.2. Hence Wedderburn theorem forces $J_k^2$ to be scalars in different $S_k$-irreps. Diagonal matrix of this form can be generated by representations of $X_1, ..., X_k$ [3, 4] and further commutes with all YJM-elements $X_1, ..., X_n$. Note that during the proof, we explicitly uses representations on computational basis and even Young basis, but as long as these operators commute on one basis, they commute on all other bases. $\square$

**Theorem VI.2.** *YJM-elements are strictly diagonal under the* $\mathrm{SU}(d)$ *irrep basis built by sequential Clebsch-Gordan decompositions. Conversely, sequentially coupled Casimir operators are strictly diagonal under the Young basis decomposed by branching rule.*

*Proof.* Let us first conduct sequential coupling and Clebsch-Gordan decomposition to construct the $\mathrm{SU}(d)$ irrep basis like [23–25]. For brevity, we first discuss the $\mathrm{SU}(2)$-case. Obviously, sequential coupling Casimir operators $J_k^2$ are strictly diagonal under the $\mathrm{SU}(2)$ total spin basis $|J; m\rangle$. Let us consider am arbitrary $\mathrm{SU}(2)$ irrep block with equivalent copies $W_\lambda^{\oplus m_{\mathrm{SU}(2),\lambda}}$ decomposed from $V^{\otimes n}$. As we are doing sequential coupling, spin labels $J$ for different copies of $W_\lambda$ are different. Since by Lemma VI.1, $X_i$ commutes with $J_k^2$, $X_i$ should be restricted to matrices on each copy of $W_\lambda$.

Now we need to show that these matrices are diagonal. This is quite straightforward when $d=2$. Since by Theorem I.2 $X_i$ commutes with the spin operator $S_z^n$ whose eigenvalues (spin-$z$ components) are nowhere identical within any $\mathrm{SU}(2)$-irrep, $X_i$ should be a scalar matrix on each $W_\lambda$ and hence diagonal. Conversely, let us consider Young basis $|\alpha_T; \mu\rangle$. By the same method, $J_k^2$ are block diagonal with respect to $(S^\lambda)^{\oplus m_{\mathrm{SU}(2),\lambda}}$. To show that they are strictly diagonal, we simply note that permutation modules $M^\mu$ classify spin components like $S_z^n$.

As explained in Section IV.A in the main text, for general $\mathrm{SU}(d)$ case, $S_z^n$ is replaced by $\frac{1}{2}d(d-1)$ operators [26] to label crowded states. By Schur-Weyl duality and Lemma VI.1, these operators commute with $X_i$. Thus proving $X_i$ is diagonal under sequential coupled $\mathrm{SU}(d)$ basis follows immediately. A subtle point occurs in the converse part, where permutation modules $M^\mu$ cannot discern isomorphic copies of $S^\lambda$ in general [4, 17]. To overcome this problem, we employ the aforementioned $\frac{1}{2}d(d-1)$ operators $(\mu_i, C_j)$ to assist classification and then apply $S_n$-branching rule to construct the Young basis. It will not cause any circular reasoning because representations of these operators are defined independent of sequential Clebsch-Gordan decomposition. A detailed procedure is:

1. Decomposing $M^\mu$ arbitrarily into $S_n$ irreps possibly with multiplicities $\bigoplus_\lambda (S^\lambda)^{\oplus m_\lambda}$.

2. By definition, $(\mu_i, C_j)$ commute with $X_k$ and hence they are block diagonal in each irrep block with copies $(S^\lambda)^{\oplus m_\lambda}$.

3. Constructing Young basis by branching rule for $S^\lambda$ and its equivalent copies. We denote them by $|\alpha_T; \mathcal{M}\rangle$ where $\mathcal{M}$ records multiplicities at present.

4. As the notation $|\alpha_T; \mathcal{M}\rangle$ indicates, YJM-elements $X_k$ are scalar matrices on $\mathrm{span}_{\mathcal{M}}\{|\alpha_T; \mathcal{M}\rangle\}$ and hence $(\mu_i, C_j)$ are block diagonal here. We just need to diagonalize them which yields the required basis.

$\square$

It makes no difference to prove the restricted universality of $\mathrm{SU}(d)$-CQA than that of $S_n$-CQA in Theorem II.7, but as a simple example emphasizing the significance of Lemma II.6, we prove the following corollary.

**Corollary VI.3.** *Let* $H_P$ *be any path-connected Hamiltonian on the* $\mathrm{SU}(2)$ *irreps basis basis, the CQA-ansatz generated by* $H_P, S_z^n, (S_z^n)^2$ *is dense in each* $\mathrm{SU}(2)$-*irrep block* $\mathrm{U}(\dim W_\lambda)$.

*Proof.* In $\mathrm{SU}(2)$ case, the $|J; m\rangle$ basis elements with known $J$ are uniquely determined by $S_z^n$ with $m$ ranging from $-j_n, ..., j_n$. Thus Lemma II.6 can be easily reshaped by replacing $X_i$ with $S_z^n$. As we now have only one operator $S_z^n$ to deal with, the proof becomes extremely simple and we present it in the following.

We need to show that four basis elements $|J; m_1\rangle, |J; m_2\rangle, |J; m_3\rangle, |J; m_4\rangle$, differences of pairs eigenvalues $(m, m^2)$ of $\mathscr{S}_z^n, (\mathscr{S}_z^n)^2$ are different in the sense from Lemma II.6. If

$$m_1 - m_2 \neq m_3 - m_4,$$

then we finish the proof. Otherwise, we still have

$$m_1^2 - m_2^2 = (m_1 - m_2)(m_1 + m_2) \neq m_3^2 - m_4^2 = (m_3 - m_4)(m_3 + m_4).$$

Because if this is no the case, we divide both sides of the above equation by $(m_1 - m_2) = (m_3 - m_4)$ (they cannot equal zero by assumption from Lemma II.6), which yields $(m_1 + m_2) = (m_3 + m_4)$. This further indicates that $m_1 = m_3, m_2 = m_4$, but this still violates the assumption. □

## VII.   PROOF OF THE CIRCUIT COMPLEXITY WITH CQA

We now prove that the necessary measurements during the gradient descent for $S_n$-CQA ansätze under the Young basis can be performed in polynomial time; while classically this problem is referred to calculating the $S_n$ Fourier coefficients whose best classical algorithms, $S_n$-fast Fourier transforms [27, 28], take factorial time. Here we assume a rather problem-independent construction of $S_n$-CQA anätze with the problem Hamiltonian:

$$H_S = (1\,2) + (2\,3) + \cdots (n-1\,n)$$

Recall that in Lemma III.1, we proved that $H_S$ is path-connected and hence $S_n$-CQA with $H_S$ is universal in any given $S_n$ irrep by Theorem II.7. Then we state the following lemma adapted from Theorem 1 in [29].

**Lemma VII.1.** *The Fourier coefficient of the $S_n$-CQA anätze per layer for any $S_n$ irrep $S^\lambda$*

$$\langle \alpha_{T'}, \mu_S | U_{\text{CQA}}(\boldsymbol{\theta}) | \alpha_T, \mu_S \rangle$$

*with*

$$U_{\text{CQA}} = \exp(-i \sum_{k,l} \beta_{kl} X_k X_l) \exp(-i\gamma H_S)$$

*can be simulated by $\mathcal{O}\left(\theta n^4 \frac{\log(\theta n^4/\epsilon)}{\log\log(\theta n^4/\epsilon)}\right)$ SWAP operators with $\theta$ being the largest absolute values of parameters.*

*Proof.* It is showed in [29] that the time evolution of a $k$-local Hamiltonian $H = \sum_i c_i \sigma_i \in \mathbb{C}[S_n]$, can be simulated by

$$\mathcal{O}\left(tCk^3 n^k \frac{\log(tCkn^k/\epsilon)}{\log\log(tCkn^k/\epsilon)}\right)$$

SWAPs where $C = \max_i |c_i|$. Our purpose is simulating a 4-local Hamiltonian $\sum_{k,l} \beta_{kl} X_k X_l$ with a 2-local one $H_S = \gamma \sum_{i=1}^{n-1}(i, i+1)$ simply at $t = 1$ with $C$ being replaced by $\max\{|\beta_{kl}|, |\gamma|\}$. Since $k = 4$, we conclude the result. □

The above lemma states that the Fourier coefficients of $S_n$-CQA ansätze can be approximated roughly in $O(n^4)$ SWAP gates in qudits, while classically there is no polynomial time algorithm known. This suggests that the $S_n$-CQA ansätze are highly unlikely to be classically tractable. To better motivate Theorem 5 in the main text, it is known that supervised quantum machine learning model can be viewed as kernel methods [30], where one can define the feature map $\phi(\mathbf{x}) \in \text{End}(H)$ for some data $\mathbf{x}$, equipped with a metric on the feature/Hilbert space $H$. In our case, the metric might be given by the problem Hamiltonian $\langle .| H_S | .\rangle = \langle \cdot, \cdot \rangle_{H_S}$, where the feature map $|\phi(\mathbf{x})\rangle$ encodes data on the Schur basis elements. Then the kernel defined as

$$\kappa(\mathbf{x}', \mathbf{x}) \equiv \langle \phi(\mathbf{x}') | U_{\text{CQA}}^\dagger(\boldsymbol{\theta}) H U_{\text{CQA}}(\boldsymbol{\theta}) | \phi(\mathbf{x}) \rangle$$

can be reduced to calculate

$$c_T c_{T'} \langle \alpha_{T'}, \mu_S | U_{\text{CQA}}^\dagger(\boldsymbol{\theta}) H U_{\text{CQA}}(\boldsymbol{\theta}) | \alpha_T, \mu_S \rangle$$

with $c_T, c_{T'}$ the coefficients resulted from encoding data $\mathbf{x}$ onto the Young basis given the irrep $S^\lambda$. The key insight of the kernel method is to perform the so-called kernel trick: one do not need ever to explicitly access the feature map $\phi$ but only its kernel, which is assumed to be classically tractable. However, Theorem 5 implies that this may not be true due to the factorial complexity of the classical $S_n$-FFT while it is polynomial on quantum circuits, suggesting a potential fruitful research direction in performing supervised quantum machine learning with $S_n$-CQA on the Young basis with rigorously-proved quantum super-exponential speedup. This would include a wide array of problems concerning with SU($d$)-symmetry or permutation equivariance. We now prove the theorem 5 from the main text.

**Theorem VII.2.** *Let $U_{\text{CQA}}^{(p)}$ denote the CQA ansätze with $p$ alternating layers and let $H \in \mathbb{C}[S_n]$ be a $SU(d)$-symmetric $k$-local Hamiltonian with most $N$ terms. Then for any $S_n$ irrep $S^\lambda$, the Fourier coefficients:*

$$\frac{\partial}{\partial \theta_\mu} \langle \alpha_{T'}, \mu_S | U_{\text{CQA}}^{(p)\dagger}(\boldsymbol{\theta}) H U_{\text{CQA}}^{(p)}(\boldsymbol{\theta}) | \alpha_T, \mu_S \rangle$$

*can be simulated in $O(pN(\theta n^4 + k^2))$ with $\theta$ being the largest absolute values of parameters.*

*Proof.* Let $H = \sum_i^N c_i h_i \in \mathbb{C}[S_n]$ with $h_i$ being $k$-local. Thus each of them can be decomposed by at most $k^2$ geometrically local SWAPs/adjacent transpositions. Suppose we have $p$ alternating layers, then for each layer,

$$\langle \alpha_{T'}, \mu_S | U_{\text{CQA}}^\dagger H U_{\text{CQA}} | \alpha_T, \mu_S \rangle = \sum_{i=1}^N c_i \langle \alpha_{T'}, \mu_S | U_{\text{CQA}}^\dagger h_i U_{\text{CQA}} | \alpha_T, \mu_S \rangle$$

Thus by the above lemma, the coefficient $\langle \alpha_{T'}, \mu_S | U_{\text{CQA}}^\dagger h_i U_{\text{CQA}} | \alpha_T, \mu_S \rangle$ can be ovulated by

$$\mathcal{O}\Big(pN\big(\theta n^4 \frac{\log(\theta n^4/\epsilon)}{\log\log(\theta n^4/\epsilon)} + k^2\big)\Big)$$

many SWAP operators on qudits within $\epsilon$ precision. $\qquad\qquad\square$

To sample matrix entries of an arbitrary $k$-local Hamiltonian $H = \sum c_i \sigma_i \in \mathbb{C}[S_n]$, we consider the following Taylor expansion to order $K$:

$$\exp(-itH) = \sum_{m=0}^K \sum_{i_1,\dots,i_m}^N \frac{(-it)^m}{m!} c_{i_1} \dots c_{i_m} \sigma_{i_1} \cdots \sigma_{i_m} + \mathcal{O}(\epsilon) = \sum_{m=0}^K \sum_{i_1,\dots,i_m}^N \widetilde{c}_{i_1\dots i_m} \widetilde{\sigma}_{i_1,\dots i_m} + \mathcal{O}(\epsilon),$$

with $N$ being the number of all $k$-local permutations and $\epsilon$ being the truncation error with $K = \mathcal{O}(\frac{\log 1/\epsilon}{\log\log 1/\epsilon})$ [29]. Thus the $\exp(-itH)$ can be approximated by the truncated series as a linear combination of matrix representations of $S_n$ group elements. In the case of $S_n$-CQA, to approximate the evolution of YJM-element Hamiltonian $\sum \beta_{kl} X_k X_l$ at $K$th order, the term $\sigma_{i_1} \cdots \sigma_{i_m}$ from above contains all possible $K$th order products of transpositions $(i,j)$ by the definition of YJM-elements. Consequently, the linear combination contains at least all $K+1$ cycles from $S_n$ whose total number scales as $n^K$ by definition [29]. Suppose we set the error $\epsilon$ to be $\mathcal{O}(\frac{1}{n})$. Assembling all the aforementioned facts, the linear combination used to approximate $S_n$-CQA contains super-polynomially many different group elements from $S_n$ and therefore is thought to be unlikely tractable by classical methods.

## VIII. RELATION TO ADIABATIC QUANTUM COMPUTING WITH SU($d$) SYMMETRY

Similar to the limiting $p \to \infty$ behavior of QAOA, our ansatz corresponds to the adiabatic quantum evolution for each $S_n$ irreducible representation subspace. The standard interpolating Hamiltonian is given by:

$$\begin{aligned}
\widehat{H}(s) &= sH_p + \sum_{k \leq \ell}(1-s)\beta_{k\ell}(s) X_k X_\ell \\
&= sH_p + (1-s)H_M(s),
\end{aligned} \tag{2}$$

where $s$ ranging from 0 to 1 is the time-parameterized path with $s(0) = 0$ and $s(T) = 1$. In contrast with the standard interpolating Hamiltonian, we set the coupling strength parameters $\{\beta_{k\ell}\}$ as path dependent I.I.D. random variables drawn from some possibly unknown distribution $\mathcal{D}$. Using the Central Limit Theorem, it suffices to draw $\{\beta_{k\ell}\}$ from normal distributions (when $n$ is large), say $\mathcal{N}(0,\sigma)$. We shall work in SU(2)-$S_n$ duality but generalization to qudits is straightforward.

Choose an irrep subspace $S^\lambda$ with the instantaneous ground and first excited states $|v_0(s)\rangle, |v_1(s)\rangle$ and the spectral gap $\Delta = \min_{s \in [0,1]} E_1(s) - E_0(s)$. The adiabatic theorem implies that $|v_0(s=T)\rangle$ is $L^2$-$\epsilon$ close to the ground state of $H_p$ if:

$$T \gg \frac{|\langle v_1(s) | \partial_s \widehat{H}(s) | v_0(s) \rangle|}{\Delta^2}, \tag{3}$$

where:

$$\partial_s \widehat{H} = H_p - \sum_{k \leq \ell} \beta_{k,\ell}(s) X_k X_\ell - s \sum_{k \leq \ell} \partial_s \beta_{k,\ell}(s) X_k X_\ell.$$

The quantum evolution in the Heisenberg picture is:

$$U(T) = \mathcal{T} \exp(-i \int_0^T dt \widehat{H}(s(t))).$$

Taken $T = N\Delta t \equiv \Delta t_N$ where $s(\Delta t_j)$ is constant at this small time increments, the standard Trotterization technique implies that:

$$
\begin{aligned}
U(T) \approx \prod_{j=1}^N &\exp(-i\Delta t_j s(\Delta t_j) H_p) \Pi_{k \leq \ell}^n \\
&\exp(-i\Delta t_j (1 - s(\Delta t_j)) \beta_{k,\ell}(\Delta t_j) X_k X_\ell).
\end{aligned}
\tag{4}
$$

Setting $\gamma_j = \Delta t_j s(\Delta t_j)$ and $(\widetilde{\beta}_{k,\ell})_j = \Delta t_j (1 - s(\Delta t_j)) \beta_{k,\ell}(\Delta t_j)$, we recover the $S_n$-CQA ansätze at $p = N$. The variational parameters $\{(\widetilde{\beta}_{k,\ell})_j\}_{k,\ell}$ for the classical optimization corresponds to the randomized coupling strength parameters for the initial Hamiltonian strictly diagonal in the total spin basis elements.

The adiabatic evolution time $T$ may be exponential if the spectral gap $\Delta$ is exponentially small as the number of spins grows. There are many works involving techniques to amplify the spectral gap, such as modifying the initial and final Hamiltonian [31, 32], quantum adiabatic brachistochrone [33], and adding non-stoquastic Hamiltonian [34, 35], to name a few. We refer interested readers to [36] for further discussion. Therefore, one may ask whether the path-dependent randomized coupling coefficients would lead to amplification of the spectral gap so that the final evolution time $T$ can be polynomial bound. In particular, one might need to consider the dynamics of the random path-dependent coupling $\beta_{k,\ell}(s)$, such as the Langevin or diffusion processes, supplemented by systems of stochastic differential equations.

The answer to this question will shed light on the existence of a quantum advantage in finding the exact eigenstate of the problem Hamiltonian. Furthermore, by relating it to the adiabatic evolution, similarly to QAOA, this would imply that the $S_n$-CQA ansätze enjoys a guaranteed performance improvement as the number of alternating layers $p$ increases. We leave this as an open question. The success in approximating the sign structure of the ground state in numerical simulation from the $S_n$-CQA ansätze may serve as heuristic evidence for such a potential quantum advantage, in addition to the super-exponential quantum speed-up in performing Fourier space convolution in quantum circuits.

## IX. DETAILS ON NUMERICAL SIMULATION

### A. benchmark Details

| List of Hyparameters used for $J_2 = 0.5$ on Rectangular Latice | | | |
|---|---|---|---|
| Hyperparameter Type | $S_n$-CQA | GCNN[37] | RBM |
| Learning Rate | 1e-2 | 2e-2 | 2e-2 |
| Number of sampling | 1000 | 1024 | 1024 |
| (Alternating) Layer | 4 | 4 | N/A |
| Optimizer | Adamw | Natural Gradient (SR) | Natural Gradient (SR) |
| Iterations | 5000 | 5000 | 5000 |
| Symmetries | SU(2) | Graph Automorphism | None |
| Feature Dimensions | N/A | (8, 8, 8, 8) | N/A |

TABLE I. List of Key hyperparameters used in Benchmarks. For a more precise list along with codes, please refer to `https://github.com/hanzheng98/Sn-CQA/tree/main/codes`

### B. Simulating with Quantum Noise

An important challenge in the near-time quantum application is to validate proposed variational ansätze in the presence of noise. We report that the $S_n$-CQA anätze are also susceptible to the noise similarly to the noise-induced barren plateau (NIBP)

| List of Hyparameters used for $J_2 = 0.8$ on Rectangular Lattice | | | |
|---|---|---|---|
| Hyperparameter Type | $S_n$-CQA | GCNN[37] | RBM |
| Learning Rate | 1e-2 | 2e-2 | 2e-2 |
| Number of sampling | 1000 | 1024 | 1024 |
| (Alternating) Layer | 4 | 4 | N/A |
| Optimizer | Adamw | Natural Gradient (SR) | Natural Gradient (SR) |
| Iterations | 5000 | 5000 | 5000 |
| Symmetries | SU(2) | Graph Automorphism | None |
| Feature Dimensions | N/A | (8, 8, 8, 8) | None |

TABLE II. List of Key hyperparameters used in Benchmarks. For a more precise list along with codes, please refer to `https://github.com/hanzheng98/Sn-CQA/tree/main/codes`

| List of Hyparameters used for $J_2 = 0.5$ on Kagome Lattice | | | |
|---|---|---|---|
| Hyperparameter Type | $S_n$-CQA | GCNN[37] | RBM |
| Learning Rate | 1e-2 | 2e-2 | 2e-2 |
| Number of sampling | 1000 | 1024 | 1024 |
| (Alternating) Layer | 6 | 4 | N/A |
| Optimizer | Adamw | Natural Gradient (SR) | Natural Gradient (SR) |
| Iterations | 5000 | 5000 | 5000 |
| Symmetries | SU(2) | Graph Automorphism | None |
| Feature Dimensions | N/A | (8, 8, 8, 8) | N/A |

TABLE III. List of Key hyperparameters used in Benchmarks. For a more precise list along with codes, please refer to `https://github.com/hanzheng98/Sn-CQA/tree/main/codes`

argued in [38]. More precisely, we assume a noise training scenario with noise-free cost function, i.e. expectation with respect to the $J_1$-$J_2$ hamiltonians, with the local noise $q_i$ such that $-1 \leq q_i \leq 1$ per $i$th site. We assume the local Pauli noise with $q_i \equiv q_i^X = q_i^Y q_i^Z$ with $\{X, Y, Z\}$ acts on the $i$th qubits. The YJM mixer Hamiltonian under the action of this type of noise decomposes as:

$$
\begin{aligned}
X_k X_l &= \sum_{i=1}^{k-1} \sum_{j=1}^{l-1} \mathbf{S}_i \cdot \mathbf{S}_k \cdot \mathbf{S}_j \cdot \mathbf{S}_l + \frac{1}{2} \mathbf{S}_i \cdot \mathbf{S}_k + \frac{1}{2} \mathbf{S}_j \cdot \mathbf{S}_l + \frac{1}{4} I \\
&\mapsto \sum_{i=1}^{k-1} \sum_{j=1}^{l-1} q_i q_j q_k q_l \mathbf{S}_i \cdot \mathbf{S}_k \cdot \mathbf{S}_j \cdot \mathbf{S}_l + q_i q_k \frac{1}{2} \mathbf{S}_i \cdot \mathbf{S}_k \\
&\quad + \frac{1}{2} q_j q_l \mathbf{S}_j \cdot \mathbf{S}_l + \frac{1}{4} I
\end{aligned}
\tag{5}
$$

We assume the noise on each site to be i.i.d random variable so the central limit theorem applies when the number of qubits are large, these noise sources can be assumed to be Gaussian. Therefore, the impact of noise on the parameters can be modelled as:

$$
\beta_{k,l} \mapsto \beta_{k,l} + \xi_{k,l}^{(1)} + \xi_{k,l}^{(2)} + \xi_{k,l}^{(3)}
$$

With $\xi_{k,l}^{(1)} = \sum_{i=1}^{k-1} \sum_{j=1}^{l-1} q_i q_j q_k q_l$, $\xi_{k,l}^{(2)} = \sum_{i=1}^{k-1} q_i q_k$, and $\xi_{k,l}^{(3)} = \sum_{j=1}^{l-1} q_j q_l$ to be Gaussian noises. We further consider the Gaussian noise $\xi_{k,l} \equiv \xi_{k,l}^{(1)} + \xi_{k,l}^{(2)} + \xi_{k,l}^{(3)}$. Similarly, we assume the parameters of the $J_1$-$J_2$ Hamiltonian:

$$
\gamma \mapsto \gamma + \zeta
$$

Where $\zeta$ is some Gaussian noise. In the numerical simulation, we model the noise in the gradient descent, where the update rule for the varitaional parameter $\theta$ is given:

$$
\theta \mapsto \theta - \eta(\partial_\theta \langle \Psi_{\text{CQA}} | H | \Psi_{\text{CQA}} \rangle + \xi)
\tag{6}
$$

This measurement-induced noise model is studied in [39] where the authors show that quantum variational ansätze in general are noise-resilient in the over-parameterized regime where the number of variational parameters scale roughly square of the dimension of the physical Hilbert space.

In addition to studying how $S_n$-CQA can be resilient to noise induced in the gradient descent, it is of interest investigate the efficient way to simulate the $S_n$-CQA dynamics (consisted of Hamiltonian evolution of problem and mixer Hamiltonians ) that

| List of Hyparameters used for $J_2 = 0.8$ on Kagome Latice | | | |
| --- | --- | --- | --- |
| Hyperparameter Type | $S_n$-CQA | GCNN[37] | RBM |
| Learning Rate | 1e-2 | 2e-2 | 2e-2 |
| Number of sampling | 1000 | 1024 | 1024 |
| (Alternating) Layer | 6 | 4 | N/A |
| Optimizer | Adamw | Natural Gradient (SR) | Natural Gradient (SR) |
| Iterations | 5000 | 5000 | 5000 |
| Symmetries | SU(2) | Graph Automorphism | None |
| Feature Dimensions | N/A | (8, 8, 8, 8) | N/A |

TABLE IV. List of Key hyperparameters used in Benchmarks. For a more precise list along with codes, please refer to `https://github.com/hanzheng98/Sn-CQA/tree/main/codes`

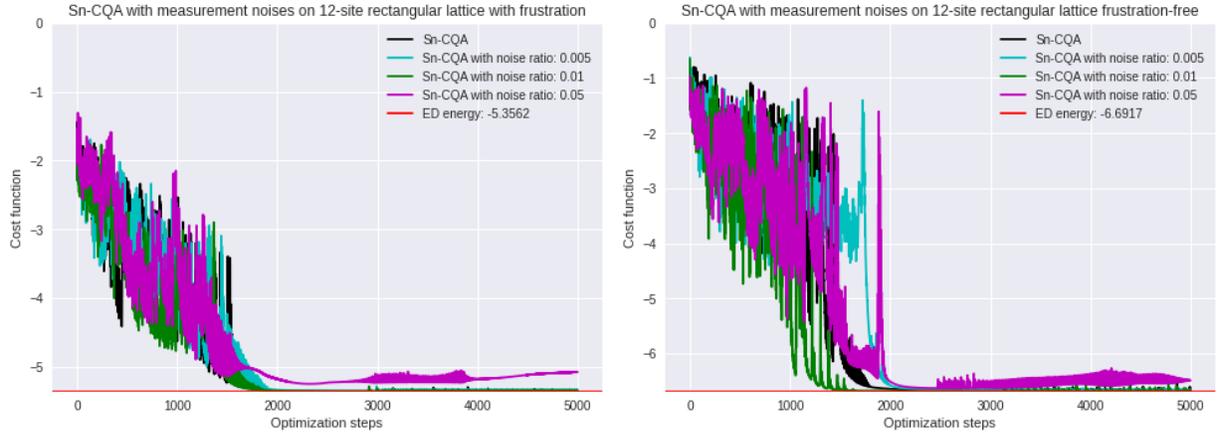

FIG. 3. For the noise model, contrary to [38], we fix $p$ and vary the magnitude of the noise $\xi$. We see that with moderate level of measurement noise, $S_n$-CQA still could converge to the ground state energy. We see with noise scale increases the performance gets worse. In any case the noise-induced barren plateau (NIBP) seems to be present if the measurement noise is large as the loss seems to be stagnant (upto to oscillating) in the later training as compared with the noise free model. We note that the frustration-free case is slightly more robust to measurement noise than frustrated case.

is robust to noise. In particular one might study how the symmetry protection protocol proposed in [40] could lead a faster Trotterization of $S_n$-CQA in NISQ era. We leave this study as intriguing future works.

---


}

\end{document}